\documentclass[final,onefignum,onetabnum]{siamart171218}

\usepackage{braket,amsfonts}

\usepackage{array}

\usepackage[caption=false]{subfig}
\usepackage{pgfplots}

\newsiamthm{claim}{Claim}
\newsiamremark{remark}{Remark}
\newsiamremark{hypothesis}{Hypothesis}
\crefname{hypothesis}{Hypothesis}{Hypotheses}

\usepackage{algorithmic}

\usepackage{graphicx,epstopdf}

\Crefname{ALC@unique}{Line}{Lines}

\usepackage{amsopn}

\usepackage{xspace}
\usepackage{bold-extra}
\usepackage[most]{tcolorbox}

\colorlet{texcscolor}{blue!50!black}
\colorlet{texemcolor}{red!70!black}
\colorlet{texpreamble}{red!70!black}
\colorlet{codebackground}{black!25!white!25}

\usepackage{multirow}
\usepackage{booktabs}
\usepackage{diagbox}
\usepackage{epsfig}
\usepackage{stmaryrd} 

\lstdefinestyle{siamlatex}{%
  style=tcblatex,
  texcsstyle=*\color{texcscolor},
  texcsstyle=[2]\color{texemcolor},
  keywordstyle=[2]\color{texemcolor},
  moretexcs={cref,Cref,maketitle,mathcal,text,headers,email,url},
}

\tcbset{%
  colframe=black!75!white!75,
  coltitle=white,
  colback=codebackground, 
  colbacklower=white, 
  fonttitle=\bfseries,
  arc=0pt,outer arc=0pt,
  top=1pt,bottom=1pt,left=1mm,right=1mm,middle=1mm,boxsep=1mm,
  leftrule=0.3mm,rightrule=0.3mm,toprule=0.3mm,bottomrule=0.3mm,
  listing options={style=siamlatex}
}

\newtcblisting[use counter=example]{example}[2][]{%
  title={Example~\thetcbcounter: #2},#1}

\newtcbinputlisting[use counter=example]{\examplefile}[3][]{%
  title={Example~\thetcbcounter: #2},listing file={#3},#1}

\DeclareTotalTCBox{\code}{ v O{} }
{ 
  fontupper=\ttfamily\color{black},
  nobeforeafter,
  tcbox raise base,
  colback=codebackground,colframe=white,
  top=0pt,bottom=0pt,left=0mm,right=0mm,
  leftrule=0pt,rightrule=0pt,toprule=0mm,bottomrule=0mm,
  boxsep=0.5mm,
  #2}{#1}

\patchcmd\newpage{\vfil}{}{}{}
\usepackage{amsmath,amssymb,mathtools}
\usepackage{graphicx,epstopdf}
\usepackage[caption=false]{subfig}
\usepackage{tikz}
\usetikzlibrary{matrix}
\DeclareMathOperator*{\argmin}{arg\,min}
\newcommand{\vep}{\varepsilon}

\newcommand{\trho}{\tilde{\rho}}

\newcommand{\bx}{{\bf x} }
\newcommand{\bz}{{\bf z} }

\newcommand{\tf}{{\tilde{f}} }
\newcommand{\tE}{{\tilde{E}} }
\newcommand{\p}{\partial}
 \def\be{\begin{equation}\displaystyle}
\def\ee{\end{equation}}
\numberwithin{equation}{section}
\newcommand{\bea}{\begin{eqnarray}}
\newcommand{\eea}{\end{eqnarray}}
\newcommand{\brho}{\boldsymbol{\rho}}
\newcommand{\tbrho}{\tilde{\boldsymbol{\rho}}}

\newsiamthm{myexample}{Example} 

\catcode`\@=11
\@addtoreset{figure}{section}
\renewcommand\thefigure{\thesection.\@arabic\c@figure}
\@addtoreset{table}{section}
\renewcommand\thetable{\thesection.\@arabic\c@table}

\begin{tcbverbatimwrite}{tmp_\jobname_header.tex}
\title{Computing ground states of Bose-Einstein Condensates with higher order interaction via a regularized density function formulation\thanks{Submitted to the journal's  {Computational Methods in Science and Engineering} section
\funding{The author's work was supported by the Academic Research
Fund of Ministry of Education of Singapore grant No.
R-146-000-247-114. }}}

\author{Weizhu Bao\thanks{Department of Mathematics, National University of
Singapore, Singapore 119076, Singapore (\email{matbaowz@nus.edu.sg}, URL: http://blog.nus.edu.sg/matbwz/).}
\and Xinran Ruan\thanks{Laboratoire J.-L. Lions, Campus UPMC, Sorbonne University, Paris, France 75005 (\email{ruan@ljll.math.upmc.fr}, \email{A0103426@u.nus.edu}).}
}

\headers{ }{ }
\end{tcbverbatimwrite}
\input{tmp_\jobname_header.tex}

\ifpdf
\hypersetup{ pdftitle={111} }
\fi

\begin{document}
\maketitle

\begin{tcbverbatimwrite}{tmp_\jobname_abstract.tex}
\begin{abstract}
We propose and analyze a new numerical method for computing the ground state of the modified Gross-Pitaevskii equation for modeling the Bose-Einstein condensate with a higher order interaction by adapting the density function formulation and the accelerated projected gradient method.
By reformulating the energy functional $E(\phi)$ with $\phi$, the wave function, in terms of the density $\rho=|\phi|^2$, the original non-convex minimization problem for defining the ground state is then reformulated to a convex minimization problem.
In order to overcome the semi-smoothness of the function $\sqrt{\rho}$ in the kinetic energy part, a regularization is introduced with a small parameter $0<\vep\ll1$.
Convergence of the regularization is established when $\vep\to0$.
The regularized convex optimization problem is discretized by the second order finite difference method.
The convergence rates in terms of the density and energy of the discretization are established.
The accelerated projected gradient method is adapted for solving the discretized optimization problem.
Numerical results are reported to demonstrate the efficiency and accuracy of the proposed numerical method.
Our results show that the proposed method is much more efficient than the existing methods in the literature, especially in the strong interaction regime.
\end{abstract}


\end{tcbverbatimwrite}
\input{tmp_\jobname_abstract.tex}

\section{Introduction}
\label{sec:intro}
The Bose-Einstein condensate (BEC) has drawn great attention since its first experimental realization in 1995 \cite{Anderson,Bradley,Davis} because of its huge  potential value in exploring a wide range of questions in fundamental physics.
Mathematically, a BEC can be effectively approximated via a Gross-Pitaevskii equation (GPE),
which is a mean field model  considering only the binary contact interaction of form \cite{Lieb,TGmath,PitaevskiiStringari}
\be\label{eq:int_gpe}
V_{\rm{int}}(\bx_1-\bx_2)=g_0\delta(\bx_1-\bx_2), \quad \bx_1,\bx_2\in\mathbb{R}^3,
\ee
where $\delta(\cdot)$ is the Dirac delta function, $g_0=\frac{4\pi \hbar^2 a_{s}}{m}$ is the contact interaction strength with $a_s$ being the $s$-wave scattering length, $\hbar$ being the reduced Planck constant and $m$ being the mass of the particle \cite{TGmath}.
The great success of GPE comes from the fact that BEC is extremely dilute and the confinement is weak in most experiments.
However, recent progress in experiments, such as the application of Feshbach resonances in cold atomic collision \cite{Chin,DErrico,Roati}, enables us to tune the scattering length in a larger range.
In such experiments, the underlying assumptions of the GPE and its validity have to be carefully examined.
As shown in  \cite{Collin,Esry,Zinner}, a higher order interaction (HOI) has to be included for a better description.
With the HOI term,  the new binary interaction  takes the form  \cite{Collin,Esry,Zinner}
\be\label{eq:int}
\tilde{V}_{\rm{int}}(\bz)=V_{\rm{int}}(\bz)+V_{\rm{HOI}}(\bz)=g_0\delta(\bz)+g_0g_1\left(\delta(\bz)\nabla^2_{\bz}+\nabla^2_{\bz}\delta(\bz)\right),
\ee
where $\bz=\bx_1-\bx_2\in\mathbb{R}^3$, $\delta(\cdot)$ and $g_0$ are defined as before, and  $g_1$, the strength of the HOI, is defined as $g_1=\frac{a_s^2}{3}-\frac{a_sr_e}{2}$ with $r_e$ being the effective range of the two-body interaction.
With the new binary interaction \eqref{eq:int}, the dimensionless modified Gross-Pitaveskii equation (MGPE) in $d$-dimensions ($d$=1,2,3) can be derived as \cite{Collin,Fu,JJG,Ruan_thesis,Ruan,Veksler}
\be\label{mgpe}
 i\partial_t \psi=\left[-\frac{1}{2}\Delta
+V(\bx)
+\beta|\psi|^2-\delta\Delta(|\psi|^2)\right]\psi,\quad t>0,\quad \bx\in\mathbb{R}^d,
 \ee
 where $\psi=\psi(\bx,t)$ is the complex-valued wave function, $V(\bx)\ge0$ is the external potential and the dimensionless parameters $\beta$ and $\delta$ denote the contact interaction strength and the HOI strength, respectively.
Define the energy to be
\be \label{energy}
\tilde{E}(\psi(\cdot,t)):=\int_{{\mathbb R}^d}\biggl[
\frac12|\nabla\psi|^2+V(\bx)|\psi|^2+\frac{\beta}{2}|\psi|^4+\frac{\delta}{2}|\nabla|\psi|^2|^2\biggl]\,d\bx.
\ee
It is easy to  check that the MGPE \eqref{mgpe}  conserves the {\sl mass} (or {\sl normalization}) \cite{Ruan_thesis} as
\be\label{prop:mass}
\|\psi(\cdot,t)\|_2^2=\int_{{\mathbb R}^d}\,|\psi(\bx,t)|^2\,d\bx \equiv \int_{{\mathbb R}^d}\,|\psi(\bx,0)|^2\,d\bx = \|\psi(\cdot,0)\|_2^2=1,
\ee
and  the {\sl energy} \eqref{energy} \cite{Ruan_thesis}, i.e.
\be
\tilde{E}(\psi(\cdot,t))\equiv \tilde{E}(\psi(\cdot,0)), \quad t\ge0.
\ee

To find the ground state of  a BEC
is a fundamental problem  since the ground state represents the most stable stationary state of the system.
Mathematically speaking, the ground state of a BEC described by the  MGPE \eqref{mgpe} is defined as the minimizer of the energy functional (\ref{energy}) under the normalization constraint \eqref{prop:mass}.
Specifically,  we have the ground state $\phi_g=\phi_g(\bx)$ defined as  \cite{mgpe-th,Ruan_thesis}
\begin{equation}\label{eq:minimize_mgpe}
    \phi_g :=  \argmin_{\phi \in S}
    \tE\left(\phi\right),
  \end{equation}
where $S:=\left\{\phi \, \left| \,\|\phi\|_2=1,\quad \tE(\phi)<\infty\right.\right\}$.
The corresponding energy $E_g:=\tE(\phi_g)$ is called the ground state energy.
The Euler-Lagrange equation of the problem \eqref{eq:minimize_mgpe} implies that
the ground state $\phi_g$ satisfies the following nonlinear eigenvalue problem \cite{mgpe-th,Ruan_thesis}
\be\label{eq:eig_mgpe}
\mu \phi=\left[-\frac{1}{2}\Delta+V(\bx)+\beta|\phi|^2-\delta\Delta(|\phi|^2)\right]\phi,\quad \bx\in\mathbb{R}^d,
\ee
under the normalization constraint $\|\phi\|_2=1$,
where $\mu$ is the corresponding chemical potential (or eigenvalue in mathematics) and can be computed as  \cite{mgpe-th,Ruan_thesis}
\begin{align}
 \label{eq:mu_mgpe}
&\mu(\phi)=\tE(\phi)+\int_{{\mathbb R}^d}\biggl[
\frac{\beta}{2}|\phi|^4+\frac{\delta}{2}\left|\nabla|\phi|^2\right|^2\biggl]\,d\bx.
\end{align}

The existence and uniqueness of the ground state of the non-convex optimization problem \eqref{eq:minimize_mgpe} has been thoroughly studied in \cite{mgpe-th}.
When $\delta=0$, the MGPE \eqref{mgpe} degenerates to a GPE and the existence and uniqueness of the ground state is clear. We refer readers to \cite{Bao2014,Bao2013,PitaevskiiStringari} and references therein.
When $\delta\neq0$, the ground state exists if and only if $\delta>0$,
and it can be chosen to be positive.
Furthermore, the positive ground state is unique if we have both $\beta\ge0$ and $\delta\ge0$.
 Therefore, throughout the paper, we will only consider the special case where $\delta\ge0$ and $\beta\ge0$  for the numerical computation. In particular, we are interested in the Thomas-Fermi limit with strong HOI effect, i.e. $\delta\gg1$.
The exact Thomas-Fermi limit of the ground state with strong repulsive interaction has been studied in details in \cite{Ruan}.

In this paper, we aim to design a numerical method to compute the ground state defined in \eqref{eq:minimize_mgpe}.
When $\delta=0$, the MGPE \eqref{mgpe} degenerates to the classical GPE and numerous methods have been proposed for this special case,
such as the normalized gradient flow method  \cite{Antoine2, Bao2004,Wz1,Chiofalo,Bao2010,Bao-Chern-Lim-2006},
a Runge-Kutta spectral method with spectral discretization in space and Runge-Kutta type integration in time  \cite{Muruganandam},
Gauss-Seidel-type methods  \cite{Chang},
a finite element method by directly minimizing the energy functional  \cite{BaoT},
a regularized Newton method  \cite{BaoWuWen},
a method based on Riemannian optimization  \cite{Danaila},
a preconditioned nonlinear conjugate gradient method  \cite{Antoine}.
However, to our best knowledge, there are few numerical schemes proposed for the case $\delta>0$.
The strong HOI effect, i.e. $\delta\gg1$, introduces high nonlinearity, which  may result in inefficiency and instability of the numerical methods proposed for the $\delta=0$ case.
For instance, the regularized Newton method would become extremely slow when $\delta$ is large,  which will be shown later via numerical experiments,
while the normalized gradient flow method could even fail as shown in \cite{Ruan_gradientflow}, even if a convex-concave splitting of the HOI term is adapted to significantly improve the robustness of the method. 

To resolve the problem, we propose in the paper a novel method, namely the \textbf{rDF-APG} method, where the term ``rDF-APG" means using the regularized density function formulation and accelerated projected gradient method for optimization.
The method works by directly minimizing the energy functional formulated via the density function $\rho=|\phi|^2$ instead of the wave function $\phi$.
In fact, the idea of formulating the energy functional via the density function has been widely used in the Thomas-Fermi models \cite{Catto,Lieb-TF} and
 the density functional theory (DFT) \cite{Cances,KS-DFT}, a most popular and versatile method in condensed-matter physics, computational physics, and computational chemistry. However, to our best knowledge, there are few works
 in adapting the density function formulation in designing efficient and accurate numerical methods for computing ground states and dynamics of BEC.
With the density function formulation, although the kinetic energy part becomes more complicated and a regularization of the part is needed to avoid possible singularity,
we gain the benefits that the optimization problem \eqref{eq:minimize_mgpe} is transformed to be a convex one
and the interaction energy terms are all in quadratic forms, the preferred form in a convex optimization problem.
Via the finite difference method to discretize the regularized density function formulation and  the accelerated projected gradient (APG) method to solve the optimization problem, we get the {rDF-APG} method.
Comparison with the regularized Newton method  shows that the method works for whatever $\beta\ge0$ and $\delta\ge0$, and is particularly suitable and efficient for the case with strong HOI effect, i.e. $\delta\gg1$.

The paper is organized as follows.
In Section \ref{sec:formulation}, we introduce the regularized density function formulation for the problem \eqref{eq:minimize_mgpe} and prove rigorously the convergence of the corresponding ground state densities as the regularization effect vanishes.
In Section \ref{sec:scheme}, we present the detailed  finite difference discretization of the regularized density function formulation
and analyze its spatial accuracy.
In Section \ref{sec:APG}, we adapt the accelerated projected gradient (APG) method for solving the finite-dimensional discretized constraint minimization problem and describe the rDF-APG method in details.
Extensive numerical experiments are provided in Section \ref{sec:numeric} to numerically verify the convergence and the spatial accuracy of the numerical ground state densities analyzed in the previous sections,
and show the efficiency of the rDF-APG method when the interaction, especially the HOI effect,  is strong.
Finally, some conclusions are drawn in Section \ref{sec:conclusion}.

\section{A regularized density function formulation} \label{sec:formulation}
In this section, we introduce the regularized density function formulation of the energy functional \eqref{energy},
and  prove the convergence of the ground states as the regularization effect vanishes.
\subsection{The formulation}
When $\beta\ge0$ and $\delta\ge0$, it is sufficient to consider the positive wave function for the ground state \cite{mgpe-th}.
Reformulating the energy functional \eqref{energy} with the density function $\rho:=|\phi|^2$, we get
\be\label{energy_rho}
E(\rho)=\int_{\mathbb{R}^d}\left[\frac{1}{2}|\nabla\sqrt{\rho}|^2+V(\bx)\rho+\frac{\beta}{2}\rho^2+\frac{\delta}{2}|\nabla\rho|^2\right]d\bx.
\ee
Then the ground state density $\rho_g=|\phi_g|^2$ can be defined as
\be\label{ground:ori}
\rho_g=\argmin_{\rho\in W} E(\rho),
\ee
where
\be\label{def:feasible_rho}
W:=\left\{\rho \, \left| \,\int_{\mathbb{R}^d} \rho(\bx)\, d\bx=1, \, \rho\ge0, \, E(\rho)<\infty \right. \right\}.
\ee
Obviously, the feasible set of $\rho$ in \eqref{ground:ori} is a simplex, which is convex, while the feasible set of $\phi$ in \eqref{energy} is a  non-convex unit sphere.
Combining with Lemma \ref{lem:convex} below, we find that the optimization problem \eqref{ground:ori}, which is formulated via the density function, is indeed a convex optimization problem if $\beta\ge0$ and $\delta\ge0$.
\begin{lemma}\label{lem:convex}
When $\beta\ge0$ and $\delta\ge0$, $E(\rho)$ \eqref{energy_rho} is convex in $W$ \eqref{def:feasible_rho}.
\end{lemma}
\begin{proof}
The proof follows the same idea as shown in \cite{Bao2013,Benguria} and thus omitted for brevity.
\end{proof}

However, the algorithm based on direct minimization of the density formulated energy functional \eqref{energy_rho} suffers severely from the fact that $|\nabla\sqrt{\rho}|$ will go unbounded as $\rho\to0^+$.
On one hand, the unboundedness of $|\nabla\sqrt{\rho}|$ would introduce uncontrollable errors after discretization.
On the other hand, the unboundedness of $|\nabla\sqrt{\rho}|$ implies the unboundedness of $\frac{\delta E}{\delta\rho}$,
which will cause the slow convergence, or even the collapse, of most efficient optimization techniques.
In fact, the above drawbacks will be numerically verified later in Section \ref{sec:numeric}.
As a result, one expects that a regularization of the kinetic term is needed and the following regularized energy functional is considered,
\begin{align}\nonumber
E^{\vep}(\rho)&=\int_{\mathbb{R}^d}\left[\frac{1}{2}|\nabla\sqrt{\rho+\vep}|^2+V(\bx)\rho+\frac{\beta}{2}\rho^2+\frac{\delta}{2}|\nabla\rho|^2\right]d\bx\\
&=\int_{\mathbb{R}^d}\left[\frac{|\nabla\rho|^2}{8(\rho+\vep)}+V(\bx)\rho+\frac{\beta}{2}\rho^2+\frac{\delta}{2}|\nabla\rho|^2\right]d\bx, \quad \rho\in W, \label{energy_regularize}
\end{align}
where $0<\vep\ll1$ is a small regularization parameter.
The convexity of $E^{\vep}(\cdot)$ can be proved in a similar way and then we get Lemma \ref{lem:convex_reg}.
\begin{lemma}\label{lem:convex_reg}
When $\beta\ge0$ and $\delta\ge0$, $E^{\vep}(\rho)$ \eqref{energy_regularize} is convex in $W$ \eqref{def:feasible_rho}.
\end{lemma}

The ground state density of the regularized energy functional $E^{\vep}(\rho)$ \eqref{energy_regularize} is defined as
\be\label{ground:reg}
\rho_g^{\vep}=\argmin_{\rho\in W} E^{\vep}(\rho).
\ee
Later in this section, we show the convergence of the corresponding ground state densities as the regularization effect vanishes, i.e. $\lim_{\vep\to0^+}\rho_g^{\vep}=\rho_g$, and characterize the convergence rate.

\subsection{Existence and uniqueness}
Before we study the convergence of $\rho_g^{\vep}$ \eqref{ground:reg}, it is necessary to show the existence and uniqueness of the ground state density $\rho_g^{\vep}$ when $\beta\ge0$ and $\delta\ge0$.
For simplicity, we introduce the function space
\be\label{def:func_space}
X_V=\left\{\rho\,\middle|\,\|\rho\|_2+\|\nabla\rho\|_2+\int_{\mathbb{R}^d}\,V(\bx)\rho\,d\bx<\infty,\,\|\rho\|_1=1,\,\rho\ge0\right\},
\ee
where $\|\rho\|_1:=\int_{\mathbb{R}^d}\,|\rho|\,d\bx$ and $V(\bx)$ satisfies  the confining condition, i.e.
\be\label{V:confine}
\lim\limits_{|\bx|\to\infty}V(\bx)=+\infty.
\ee
%
%
We start with the following lemma.
\begin{lemma}\label{lem:embed}
Suppose the external potential $V(\bx)$ satisfies  the confining condition \eqref{V:confine}, then for any sequence $\{\rho_n\}\subset X_V$ such that $\int_{\mathbb{R}^d}V(\bx)\rho_n\le C$ for some constant $C$ uniformly  and $\rho_n\rightharpoonup\rho$ in $H^1(\mathbb{R}^d)$, we have $\rho\ge0$, $\|\rho\|_1=1$ and
$\rho_n\to\rho$ in $L^1(\mathbb{R}^d)$.
\end{lemma}
\begin{proof}
For any $\eta>0$,  the confining condition of $V(\bx)$ \eqref{V:confine} indicates that we can choose $R$ large enough such that $V(\bx)\ge\frac{4C}{\eta}$ for all $\bx\in\Omega_R^c:=\mathbb{R}^d\backslash\Omega_R$, where $\Omega_R=\{\bx\,|\, |\bx|<R\}$. Then $\int_{\Omega_R^c}\rho_n\,d\bx$ is uniformly bounded by $\frac{\eta}{4}$ since
\be\label{proof1}
\int_{\Omega_R^c}\rho_n\,d\bx\le\frac{\eta}{4C}\int_{\Omega_R^c}V(\bx)\rho_n\,d\bx \le\frac{\eta}{4C}\int_{\mathbb{R}^d}V(\bx)\rho_n\,d\bx\le\frac{\eta}{4}.
\ee
In the bounded domain $\Omega_R$,  the weak convergence in $H^1$-norm implies $\rho_n\to\rho$ in $L^2$-norm due to the Sobolev embedding theorem and, further, $\rho_n\to\rho$ in $L^1(\Omega_R)$ by the H\"{o}lder's inequality.
Therefore, for all $n$ large enough, we have
\be\label{proof2}
\int_{\Omega_R}|\rho_n-\rho|\,d\bx\le\frac{\eta}{2}.
\ee
Besides, we can choose a subsequence of $\{\rho_n\}$, denoted as the sequence itself for simplicity, which converges pointwisely to $\rho$ in $\Omega_R$.
By Fatou's lemma,
\be\label{proof3}
\int_{\Omega_R}\rho\,d\bx\le \liminf_{n\to\infty}\int_{\Omega_R}\rho_n\,d\bx\le C.
\ee
Noticing that \eqref{proof3} can be extended to any bounded domain, we get $\int_{\mathbb{R}^d}\,\rho\,d\bx\le C$.
Following a similar argument as in \eqref{proof1}, we have
\be\label{proof4}
\int_{\Omega_R^c}\rho\,d\bx\le\frac{\eta}{4}.
\ee
Combining \eqref{proof1}, \eqref{proof2} and \eqref{proof4}, we have that, for any given $\eta>0$,
$
\int_{\mathbb{R}^d}\,|\rho-\rho_n|\,d\bx\le\eta
$
if $n$ is large enough, which indicates $\rho_n\to\rho$ in $L^1(\mathbb{R}^d)$ and $\|\rho\|_1=1$.
The conclusion $\rho\ge0$ comes from the pointwise convergence of a subsequence of $\{\rho_n\}$. The details are omitted here for brevity.
\end{proof}

With Lemma \ref{lem:embed}, we are ready to show the existence and uniqueness of the ground state density of the regularized energy functional.
When $\delta=0$ and $\beta\ge0$, it is obvious to have the existence and uniqueness of the ground state density.
Here we focus on the case with $\delta>0$ and an arbitrary $\beta\in\mathbb{R}$, and get the following theorem.
\begin{theorem}\label{thm:mres}(Existence and uniqueness)
Assume that $\delta>0$ and $V(\bx)\ge 0$ satisfies the confining condition \eqref{V:confine},
 then  there exists a ground state density $\rho_g^{\vep}$ of the regularized energy functional $E^{\vep}(\cdot)$ \eqref{energy_regularize}  in $W$ \eqref{def:feasible_rho}.
And the ground state density is unique if we further assume $\beta\ge0$.
\end{theorem}
\begin{proof}
To show the existence, we start with the following inequality
\be\label{proof:exist1}
\int_{\Bbb R^d}|\rho|^2\,d\bx\leq \left(\tilde{C}\int_{\Bbb R^d}\rho(\bx)\,d\bx\right)^{4/d+2}\|\nabla\rho\|_2^{2d/d+2}
\leq  \frac{\tilde{C}}{\eta}+\eta\|\nabla\rho\|_2^2,
\ee
for any $\eta>0$,
where the first inequality is due to the Nash inequality and the second one is due to the Young's inequality. When $\beta\ge0$, it is easy to see that
$E^{\vep}(\rho)$ is bounded from below. When $\beta<0$,
 by choosing $\eta$ such that $\eta|\beta|=\frac{\delta}{2}$, the inequality \eqref{proof:exist1} immediately indicates that $E^\vep(\rho)$ is bounded from below since
\begin{equation}
E^{\vep}(\rho)\ge \int_{\Bbb R^d}\left(\frac12|\nabla\sqrt{\rho+\vep}|^2+V(\bx)\rho+\frac{\delta}{4}\left|\nabla\rho\right|^2\right)\,d\bx-\frac{\tilde{C}\beta^2}{\delta}.
\end{equation}
Taking a minimizing sequence $\{\rho_n\}$ in $W$,  then $\{\rho_n\}$ is uniformly bounded in $X_V$ with a weak limit $\rho_\infty$ in $H^1$-norm.
Following a similar procedure as shown in \cite{mgpe-th} and applying Lemma \ref{lem:embed}, we can show that $\rho_\infty$ is indeed the ground state density. The details are omitted here for brevity.

If we further assume that $\beta\ge0$, the uniqueness of the ground state density comes from the convexity of the regularized energy functional $E^\vep(\rho)$ in $\rho$.
\end{proof}

\subsection{Convergence of the ground states}\label{subsec:gamma_conv}
In this section, we show the convergence of $\rho_g^\vep$ as $\vep\to0^+$,
where $\rho_g^\vep$ is the ground state density of $E^{\vep}(\cdot)$ \eqref{energy_regularize}.
To begin with, we show the monotonicity of $E^{\vep}(\cdot)$ as $\vep\to0^+$, which is simple but important in the proof of the convergence.
\begin{lemma}\label{lemma:E_vep}
For any fixed state $\rho\in X_V\eqref{def:func_space}$ and $\vep_1\ge\vep_2\ge0$, we have
\be
E^{\vep_1}(\rho)\le E^{\vep_2}(\rho).
\ee
\end{lemma}
\begin{proof}
Noticing that $|\nabla\sqrt{\rho+\vep_1}|^2\le|\nabla\sqrt{\rho+\vep_2}|^2$  holds true
for any $\rho\ge0$ and $\vep_1\ge\vep_2\ge0$, it is obvious that
\be\label{proof:E_monotone}
E^{\vep_1}(\rho)\le E^{\vep_2}(\rho).
\ee
\end{proof}
In particular, when $\beta\ge0$ and $\delta\ge0$, taking $\vep_2=0$ and $\rho=\rho_g$ \eqref{ground:ori} in \eqref{proof:E_monotone} and recalling the definition of $\rho_g^{\vep}$ \eqref{ground:reg}, we have that
\be
 E^{\vep}(\rho_g^{\vep})\le E^{\vep}(\rho_g) \le E(\rho_g), \text{ for any } \vep\ge0.
\ee

With Lemma \ref{lemma:E_vep},  we are ready to show the convergence of $\rho_g^{\vep}$ as $\vep\to0^+$. The result is summarized as the following theorem.
\begin{theorem}\label{thm:gm_conv}
For ground state densities $\rho_g$ and $\rho_g^{\vep}$ defined in \eqref{ground:ori} and \eqref{ground:reg} with $\beta>0$ and $\delta>0$,
 respectively, we have
$\rho_g^{\vep}\to \rho_g$ in $H^1(\mathbb{R}^d)$ and
\be
\lim_{\vep\to0^+}E^{\vep}(\rho_g^{\vep})=E(\rho_g).
\ee
\end{theorem}
\begin{proof}
The conclusion comes directly from the $\Gamma$-convergence of 
 $E^\vep(\cdot)$ as $\vep\to 0^+$.
In fact, it is sufficient to show the lower semi-continuity of $E^\vep(\cdot)$ for all $\vep\ge0$ in $X_V$ \eqref{def:func_space} with respect to the norm
\be
\|\rho\|_{V} := \int_{\mathbb{R}^d}\,V(\bx)|\rho|\,d\bx + \|\rho\|_{H^1}
\ee
since Lemma \ref{lemma:E_vep} indicates that
the energy functional $E^\vep(\cdot)$ is monotonely increasing as $\vep\to 0^+$ \cite{Andrea_Gamma_conv}.
Here we prove the lower semi-continuity of $E^\vep(\cdot)$  via definition.
To be more specific, for any convergent sequence $\{\rho_n\}\subset X_V$ \eqref{def:func_space} satisfying
\be\label{proof:assump_Xv-conv}
\lim_{n\to\infty} \, \|\rho_n-\rho\|_V = 0,
\ee
for some $\rho\in X_V$, we aim to show that
\be\label{proof:lsc_E}
\liminf_{n\to\infty} E^{\vep}(\rho_n) \ge E^{\vep}(\rho).
\ee

To begin with, we claim that
\be\label{proof:sqrt_rho_weak_conv}
\nabla\sqrt{\rho_n+\vep} \rightharpoonup \nabla\sqrt{\rho+\vep} \quad \text{ in }
 L^2(\mathbb{R}^d).
\ee
Recalling the definition of weak limit, we only need to show that
\be\label{proof:sqrt_rho}
\lim_{n\to\infty}\int_{\mathbb{R}^d}\nabla
\sqrt{\rho_n+\vep}\,\phi\,d\bx=\int_{\mathbb{R}^d}\nabla\sqrt{\rho+\vep}\,\phi\,d\bx,
\ee
for any compactly supported smooth function $\phi$.
Here we assume $\rm{supp}(\phi)\subset\Omega_R:=\{\bx\,|\, |\bx|<R\}$ for some $R>0$.
Obviously, $\rho_n\to \rho$ in $L^1(\Omega_R)$ by combining \eqref{proof:assump_Xv-conv} and the H\"{o}lder inequality.
A direct computation shows that
\be\label{proof:sqrt_rho1}
\limsup_{n\to\infty}\int_{\Omega_R}\left|
\sqrt{\rho_n+\vep}-\sqrt{\rho+\vep}\right|\,d\bx
\le
\limsup_{n\to\infty}\int_{\Omega_R}\frac{|\rho_n-\rho|}{2\sqrt{\vep}}\,d\bx= 0,
\ee
which implies \eqref{proof:sqrt_rho}, and thus \eqref{proof:sqrt_rho_weak_conv}, immediately by combining the fact that
\begin{align}
\int_{\mathbb{R}^d}\left(\nabla\sqrt{\rho_n+\vep}-\nabla\sqrt{\rho+\vep}\right)\phi\,d\bx=-\int_{\Omega_R}\left(\sqrt{\rho_{n}+\vep}-\sqrt{\rho+\vep}\right)\nabla\phi\,d\bx,
\end{align}
and the fact that $\phi$ is smooth and compactly supported within $\Omega_R$.

By the weak convergence \eqref{proof:sqrt_rho_weak_conv}, we have that
\be\label{proof:weak2}
\liminf_{n\to\infty}\|\nabla\sqrt{\rho_n+\vep}\|_2\ge\|\nabla\sqrt{\rho+\vep}\|_2.
\ee
On the other hand, the convergence of $\{\rho_n\}$ in the sense \eqref{proof:assump_Xv-conv} implies that
\be
\lim_{n\to\infty}\int_{\mathbb{R}^d}V(\bx)\rho_n\,d\bx = \int_{\mathbb{R}^d}V(\bx)\rho\,d\bx,
\ee
and
\be \label{proof:weak1}
\lim_{n\to\infty}\|\rho_n\|_2=\|\rho\|_{2},\quad
\lim_{n\to\infty}\|\nabla\rho_n\|_2=\|\nabla\rho\|_2.
\ee
Thus we proved the lower semi-continuity of $E^\vep(\cdot)$ \eqref{proof:lsc_E}. \
\end{proof}

\begin{remark}
One key step \eqref{proof:sqrt_rho1} in the proof is no longer valid when $\vep=0$.
However, we can still prove the lower semicontinuity of the energy functional $E^{0}(\cdot)$ via a similar, but more complicated argument.
\end{remark}

Theorem \ref{thm:gm_conv} shows the convergence of $\rho_g^\vep$.
Although the exact convergence rate is unclear,
the theorem below shows a relation between the convergence rate of the ground state density $\rho_g^\vep$ \eqref{ground:reg} and the convergence rate of the corresponding energy.

\begin{theorem}\label{thm:conv_rate}
For $\rho_g$ and $\rho_g^{\vep}$ defined in \eqref{ground:ori} and \eqref{ground:reg} with $\beta\ge0$ and $\delta\ge0$, respectively,  we have
\be\label{conv_speed}
\frac{\beta}{2}\|\rho_g-\rho_g^{\vep}\|_2^2+\frac{\delta}{2}\|\nabla(\rho_g-\rho_g^{\vep})\|_2^2\le E(\rho_g)-E^{\vep}(\rho_g^{\vep}).
\ee
\end{theorem}
\begin{proof}
We start by proving the following lemma, which plays an essential role  in the proof of \eqref{conv_speed}.
\begin{lemma}\label{lemma:residual_cts}
Define $
R(\rho_g,\rho_g^{\vep})=\int_{\mathbb{R}^d} \left(
\beta\rho_g^{\vep}-\delta\Delta\rho_g^{\vep}\right)(\rho_g-\rho_g^{\vep})\,d\bx
$ and  $I(\rho_g,\rho_g^{\vep})=\int_{\mathbb{R}^d}\,V(\bx)(\rho_g-\rho_g^{\vep})\,d\bx$, then
\be
R(\rho_g,\rho_g^{\vep})+I(\rho_g,\rho_g^{\vep})
\ge \int_{\mathbb{R}^d} \left[-\frac{\nabla\rho_g^{\vep}\cdot\nabla(\rho_g-\rho_g^{\vep})}{4(\rho_g^{\vep}+\vep)}
+\frac{|\nabla\rho_g^{\vep}|^2(\rho_g-\rho_g^{\vep})}{8(\rho_g^{\vep}+\vep)^2}\right]d\bx.
\label{ineq:residual_cts}
\ee
\end{lemma}
To prove Lemma \ref{lemma:residual_cts},
we consider $f(t)=E^{\vep}(\rho_g^{\vep}+t(\rho_g-\rho_g^{\vep}))$. Obviously, $\rho_g^{\vep}+t(\rho_g-\rho_g^{\vep})\in W$ \eqref{def:feasible_rho} for any $t\in[0,1]$. Therefore, $f(t)$ takes its minimum value in $[0,1]$ at $t=0$ by recalling the definition of $\rho_g^{\vep}$ \eqref{ground:reg}, which indicates $f'(0)\ge0$. A direct computation of $f'(0)$ shows that the requirement $f'(0)\ge0$ is exactly \eqref{ineq:residual_cts}. The details are omitted here for brevity.

Now we come back to the proof of \eqref{conv_speed}.
For simplicity, we define
\be
K(\rho_g,\rho_g^{\vep})=\int_{\mathbb{R}^d}\left[\frac{1}{2}|\nabla\sqrt{\rho_g}|^2-\frac{1}{2}|\nabla\sqrt{\rho_g^{\vep}+\vep}|^2\right]\, d\bx,
\ee
and
\be
D(\rho_g,\rho_g^{\vep})=\frac{\beta}{2}\|\rho_g-\rho_g^{\vep}\|_2^2+\frac{\delta}{2}\|\nabla(\rho_g-\rho_g^{\vep})\|_2^2.
\ee
Then, by a direct computation, we have that
\be
\int_{\mathbb{R}^d}\left[\frac{\beta}{2}(|\rho_g|^2-|\rho_g^{\vep}|^2)+\frac{\delta}{2}(|\nabla\rho_g|^2-|\nabla\rho_g^{\vep}|^2)\right]d\bx
=R(\rho_g,\rho_g^{\vep})+D(\rho_g,\rho_g^{\vep}),
\ee
and, therefore,
\begin{align}\nonumber
E(\rho_g)-E^{\vep}(\rho_g^{\vep})&= K(\rho_g,\rho_g^{\vep}) +I(\rho_g,\rho_g^{\vep})+R(\rho_g,\rho_g^{\vep})+D(\rho_g,\rho_g^{\vep}) \nonumber\\
\nonumber&\ge K(\rho_g,\rho_g^{\vep})\\
&\quad\nonumber+\int_{\mathbb{R}^d}\left[-\frac{\nabla\rho_g^{\vep}\cdot\nabla(\rho_g-\rho_g^{\vep})}{4(\rho_g^{\vep}+\vep)}
+\frac{|\nabla\rho_g^{\vep}|^2(\rho_g-\rho_g^{\vep})}{8(\rho_g^{\vep}+\vep)^2}\right]\, d\bx+D(\rho_g,\rho_g^{\vep})\label{proof:conv_rate_step2}\\
\nonumber&=\int_{\mathbb{R}^d}\left[\frac{|\nabla\rho_g|^2}{8\rho_g}+\frac{|\nabla\rho_g^{\vep}|^2(\rho_g+\vep)}{8(\rho_g^{\vep}+\vep)^2}-\frac{\nabla\rho_g^{\vep}\cdot\nabla\rho_g}{4(\rho_g^{\vep}+\vep)}
\right]\, d\bx+D(\rho_g,\rho_g^{\vep})\\
&\ge\int_{\mathbb{R}^d}\left[\frac{|\nabla\rho_g||\nabla\rho_g^{\vep}|}{4(\rho_g^{\vep}+\vep)}-\frac{\nabla\rho_g^{\vep}\cdot\nabla\rho_g}{4(\rho_g^{\vep}+\vep)}
\right]\, d\bx+D(\rho_g,\rho_g^{\vep})\ge D(\rho_g,\rho_g^{\vep}),
\end{align}
where the first inequality is due to Lemma \ref{lemma:residual_cts}, the last two inequalities are by Young's inequality and Cauchy-Schwarz inequality, respectively.
\end{proof}

\begin{remark}
The error estimate in Theorem \ref{thm:conv_rate} is sub-optimal.
Later in Section \ref{sec:numeric}, we will show that, in some specific  setting, we observe the
asymptotic error of order $O(\vep)$ in both energy and the $L^2$-norm.
However, only the $O(\sqrt{\vep})$ error in $L^2$-norm can be expected via Theorem \ref{thm:conv_rate}.
\end{remark}

\section{Numerical discretization}\label{sec:scheme}
In this section, we present the detailed spatial discretization of the regularized energy functional $E^\vep(\cdot)$ \eqref{energy_regularize}.
The second-order finite difference method is applied to discretize the energy functional in the regularized density function formulation  and the spatial accuracy of the corresponding numerical ground state density is analyzed in details.

\subsection{A second-order finite difference discretization}\label{sec:FD}
Due to the exponential decay of the ground state under the confining external potential \eqref{V:confine} \cite{mgpe-th}, it is reasonable to truncate the problem on a bounded, but sufficiently large domain with the homogeneous Dirichlet  boundary condition imposed.
For simplicity, only the one dimension (1D) case is considered here. Extensions to higher dimensions for tensor grids  are straightforward.

In 1D, truncate the problem into an interval $U = (a, b)$ and choose the uniform grid points as
\be\label{notation:x}
x_j = a + jh, \text{ for } j =  0, 1,\dots, N,
\ee
with mesh size $h = (b - a)/N$.
Let $\rho_j$ be the numerical approximation of $\rho(x_j)$ and denote $\brho_{h}$ to be the  vector of length $N-1$ with its $j$-th component to be $\rho_j$, i.e.
\be\label{notation:rho_full}
\brho_{h} = (\rho_1,\dots , \rho_{N-1})^T\in\mathbb{R}^{N-1}.
\ee
By the homogenous Dirichlet boundary condition, we have $\rho_0=\rho_N=0$.
Define the discrete norms of the vector $\brho_{h}$ as
\bea\label{def:norm}
\quad \|\brho_{h}\|_{l^1}:=h\sum_{j=1}^{N-1}|\rho_j|, \ \|\brho_{h}\|_{l^2}:=\sqrt{h\sum_{j=1}^{N-1}|\rho_j|^2}, \ \|\delta_x^+\brho_{h}\|_{l^2}:=\sqrt{h\sum_{j=0}^{N-1}|\delta_x^+ \rho_j|^2},
\eea
where $\delta_x^+$, which is the  finite difference approximation of $\p_x$, is defined as
$\delta_x^+\rho_{j}:=(\rho_{j+1}-\rho_j)/h$.
Then the feasible set $W_h$ can be defined as
\be\label{def:feasible}
W_h=\{\brho_{h}=(\rho_1,\dots,\rho_{N-1})^T\in\mathbb{R}^{N-1}\,|\,\|\brho_{h}\|_{l^1}=1, \brho_{h}\ge0\}.
\ee

With the notations defined above, we are now ready to  discretize the regularized energy functional $E^\vep(\cdot)$ \eqref{energy_regularize} in 1D.
Noticing the boundary condition $\rho_0=\rho_N=0$, we have
\begin{align}
E^{\vep}(\rho)
\nonumber&=\int_a^b\left[\frac{|\p_x\rho|^2}{8(\rho+\vep)}+V(x)\rho+\frac{\beta}{2}\rho^2+\frac{\delta}{2}|\p_x\rho|^2\right]dx\\
\nonumber&=\int_a^b\left[\frac{|\p_x\rho|^2}{8(\rho+\vep)}+\frac{\delta}{2}|\p_x\rho|^2\right]dx+\int_a^b\left[V(x)\rho+\frac{\beta}{2}\rho^2\right]dx\\
\nonumber&\approx
h\sum_{j=0}^{N-1}\left[\frac{1}{8}\frac{|\delta_x^+\rho_j|^2}{\rho_{j+\frac12}+\vep}+\frac{\delta}{2}\left|\delta_x^+\rho_{j}\right|^2\right]+
h\sum_{j=0}^{N-1}\left[V(x_j)\rho_j+\frac{\beta}{2}\rho_j^2\right]
\\
&\approx
h\sum_{j=0}^{N-1}\left[\frac{1}{4}\frac{|\delta_x^+\rho_j|^2}{\rho_j+\rho_{j+1}+2\vep}+V(x_j)\rho_j+\frac{\beta}{2}\rho_j^2+\frac{\delta}{2}\left|\delta_x^+\rho_{j}\right|^2\right], \label{discretization:FD_derive_end}
\end{align}
where we approximate the first integral on the second line via the composite midpoint rule and the second integral via the composite trapezoidal rule.
Define
\be\label{Eh_rho}
E_h^{\vep}(\brho_{h})=h\sum_{j=0}^{N-1}\left[\frac{1}{4}\frac{|\delta_x^+\rho_j|^2}{\rho_j+\rho_{j+1}+2\vep}+V(x_j)\rho_j+\frac{\beta}{2}\rho_j^2+\frac{\delta}{2}\left|\delta_x^+\rho_{j}\right|^2\right],
\ee
where $\brho_{h}\in W_h$ \eqref{def:feasible} and $\rho_0=\rho_N=0$.
The problem \eqref{ground:reg} now becomes finding $\brho_{g,h}^{\vep}=(\rho_{g,1}^{\vep},\rho_{g,2}^{\vep},\dots,\rho_{g,N-1}^{\vep})^T\in W_h$ such that
\be\label{problem:E_rho_FD}
\brho_{g,h}^{\vep}=\argmin_{\brho_{h}\in W_h} E_h^{\vep}(\brho_{h}).
\ee
And the corresponding numerical energy $E_{g,h}^{\vep}:=E_h^{\vep}(\brho_{g,h}^{\vep})$ approximates $E^\vep(\rho_g^\vep)$.

By Theorem \ref{thm:gm_conv}, the ground state density $\rho_g^\vep$ of the regularized energy functional \eqref{energy_regularize} will converge to the ground state density $\rho_g$ of \eqref{energy_rho} as $\vep\to0^+$.
Therefore, to solve the original optimization problem \eqref{ground:ori},
we only need to choose a sufficiently small $\vep$ and find the corresponding  numerical ground state density $\brho_{g,h}^{\vep}$ as an approximation.

We claim that the problem \eqref{problem:E_rho_FD} is a finite-dimensional convex optimization problem.
The convexity of the feasible set $W_h$ \eqref{def:feasible} is obvious.
And the convexity of $E_h^{\vep}(\cdot)$ is shown in Lemma \ref{Eh_convex} below.
\begin{lemma}\label{Eh_convex}
If $\beta\ge0$ and $\delta\ge0$,
the discretized energy functional
$E_h^{\vep}(\cdot)$  \eqref{Eh_rho} is convex  in the feasible set $W_h$ \eqref{def:feasible}  for any fixed $\vep\ge0$.
\end{lemma}
\begin{proof}
The convexity of the external potential term and the last two quadratic terms in \eqref{Eh_rho} is obvious.
Therefore, we only need to show the convexity of the kinetic energy term, i.e. the first term,  in \eqref{Eh_rho}.
Define $E^{\rm{kin}}(\brho_{h})=\sum_{j=0}^{N-1}E^{\rm{kin}}_{j}(\brho_{h})$, where
\be\label{def:kinetic_j}
E^{\rm{kin}}_{j}(\brho_{h}):=\frac{1}{4h^2}\frac{|\rho_{j+1}-\rho_j|^2}{\rho_j+\rho_{j+1}+2\vep}.
\ee
For any two vectors $\mathbf{a}_h=(a_1,\dots,a_{N-1})^T\in W_h$ and $\mathbf{b}_h=(b_1,\dots,b_{N-1})^T\in W_h$ with $a_0=a_N=b_0=b_N=0$, a direct computation shows that, for any  $j=0,1,\dots, N-1$,
\begin{align}
\nonumber&\frac{E^{\rm{kin}}_{j}(\mathbf{a}_h)+E^{\rm{kin}}_{j}(\mathbf{b}_h)}{2}=\frac{1}{8h^2}\left(\frac{|a_{j+1}-a_j|^2}{a_j+a_{j+1}+2\vep}+\frac{|b_{j+1}-b_j|^2}{b_j+b_{j+1}+2\vep}\right)\\
\nonumber&=\frac{1}{8h^2}\frac{|a_{j+1}-a_j|^2\left(1+\frac{b_j+b_{j+1}+2\vep}{a_j+a_{j+1}+2\vep}\right)+|b_{j+1}-b_j|^2\left(1+\frac{a_j+a_{j+1}+2\vep}{b_j+b_{j+1}+2\vep}\right)}{(a_j+a_{j+1}+2\vep)+(b_j+b_{j+1}+2\vep)}\\
\nonumber&\ge\frac{1}{8h^2}\frac{(a_{j+1}-a_j)^2+(b_{j+1}-b_j)^2+2(a_{j+1}-a_j)(b_{j+1}-b_j)}{(a_j+a_{j+1}+2\vep)+(b_j+b_{j+1}+2\vep)}\\
&=\frac{1}{8h^2}\frac{|a_{j+1}+b_{j+1}-a_j-b_j|^2}{(a_j+a_{j+1}+2\vep)+(b_j+b_{j+1}+2\vep)}
=E^{\rm{kin}}_{j}\left(\frac{\mathbf{a}_h+\mathbf{b}_h}{2}\right),
\end{align}
 where the inequality is due to the Cauchy's inequality,
which shows the convexity of $E^{\rm{kin}}_{j}(\cdot)$.
The convexity of $E^{\rm{kin}}(\cdot)$ is then obvious.
\end{proof}
The existence and uniqueness of the ground state density $\brho_{g,h}^{\vep}$ \eqref{problem:E_rho_FD} for a fixed $\vep\ge0$ when $\beta\ge0$ and $\delta\ge0$ come directly from the fact that the convex function $E_h^{\vep}(\cdot)$ is bounded below by 0.

To solve the problem \eqref{problem:E_rho_FD},
it is usually important to get the gradient of the discretized energy functional $E_h^{\vep}(\brho_{h})$ since most efficient optimization techniques are gradient-based.
Denote
\be
\mathbf{g}_h^{\vep}(\brho_{h})=\left( \frac{\p E_h^{\vep}}{\p\rho_1},\frac{\p E_h^{\vep}}{\p\rho_2},\dots,\frac{\p E_h^{\vep}}{\p\rho_{N-1}}\right)^T.
\ee
Then its $j$-th component can be explicitly computed  as
\be\label{Eh_rho_gradient}
\mathbf{g}_h^{\vep}[j]:=\frac{\partial E_h^{\vep}}{\partial \rho_j}=h\left[-\frac{\delta_x^+f_{j-1}}{2}-\frac{f_{j-1}^2+f_j^2}{4}+V(x_j)+\beta \rho_j-\delta(\delta_x^2\rho_j)\right],
\ee
where $\delta_x^2\rho_j:=(\rho_{j+1}-2\rho_j+\rho_{j-1})/h^2$ for $j=1,2,\ldots,N-1$ and
\be\label{def:s_f}
 f_j=f_j(\brho_{h}):=\frac{\delta_x^+\rho_j}{\rho_j+\rho_{j+1}+2\vep},\quad \text{ for }  j=0, 1, \dots, N-1.
\ee

The energy functional $E_h^{\vep}(\brho_{h})$ \eqref{Eh_rho} and its gradient $\mathbf{g}_h^{\vep}(\brho_{h})$ \eqref{Eh_rho_gradient} can be written in a more concise form.
Define a symmetric matrix $A=(a_{jk})\in\mathbb{R}^{(N-1)\times(N-1)}$, where 
\bea \label{matrix:Delta}
a_{jk}=\frac{1}{h^2}\left\{\begin{array}{lll}
2,  &\text{ if } j=k,\\
-1,  &\text{ if } j=k\pm1,\\
0, &\text{ otherwise.}
\end{array}\right.
\eea
Then we have $E_h^{\vep}(\brho_{h})$ \eqref{Eh_rho} and $\mathbf{g}_h^{\vep}(\brho_{h})$ \eqref{Eh_rho_gradient}  in  compact forms as
\begin{align}\label{Eh_rho_matrix}
&E_h^{\vep}(\brho_{h})=h E^{\rm{kin}}(\brho_{h})+h\left[\mathbf{v}^T\brho_{h}+\frac{\beta}{2} \brho_{h}^T\brho_{h}+\delta \brho_{h}^TA\brho_{h}\right],\\
\label{Eh_rho_gradient_matrix}
&\mathbf{g}_h^{\vep}(\brho_{h})=h\nabla E^{\rm{kin}}(\brho_{h})+h\left[\mathbf{v}+\beta \brho_{h}+2\delta A\brho_{h}\right],
\end{align}
respectively, where $E^{\rm{kin}}(\brho_{h})=\sum_{j=0}^{N-1}E^{\rm{kin}}_{j}(\brho_{h})$ with $E^{\rm{kin}}_{j}(\brho_{h})$  defined in \eqref{def:kinetic_j},  $\mathbf{v}=(V(x_1), V(x_2),\dots,V(x_{N-1}))^T\in\mathbb{R}^{N-1}$ and
\be
\nabla E^{\rm{kin}}(\brho_{h})=\left(\frac{\p E^{\rm{kin}}(\brho_{h})}{\p \rho_1}, \frac{\p E^{\rm{kin}}(\brho_{h})}{\p \rho_2},\dots,\frac{\p E^{\rm{kin}}(\brho_{h})}{\p \rho_{N-1}}\right)^T\in\mathbb{R}^{N-1},
\ee
with
\be
\frac{\p E^{\rm{kin}}(\brho_{h})}{\p \rho_j}=-\frac{\delta_x^+f_{j-1}}{2}-\frac{f_{j-1}^2+f_j^2}{4},
\ee
where $f_j=f_j(\brho_{h})$ is defined in \eqref{def:s_f}.

\subsection{Spatial error analysis}\label{subsec:FD_conv}
In this section, we study the spatial accuracy of the numerical ground state density. For simplicity, only the 1D case is considered here.
The following theorem can be proved, which shows that the difference between the ground state density $\brho_{g,h}^{\vep}$ \eqref{problem:E_rho_FD} and an arbitrary vector $\brho_{h}\in W_h$  \eqref{def:feasible} is closely related to the difference of the corresponding energies.
\begin{theorem}\label{thm:conv_h_E}
Considering the ground state density $\brho_{g,h}^{\vep}=(\rho^{\vep}_{g,1},\rho^{\vep}_{g,2},\dots,\rho^{\vep}_{g,N-1})^T$ \eqref{problem:E_rho_FD} and an arbitrary  vector $\brho_{h}=(\rho_1,\rho_2,\dots,\rho_{N-1})^T$ in  $W_h$ \eqref{def:feasible}, we have
\be\label{thm:conv_h_E1}
\frac{\beta}{2}\|\brho_{h}-\brho_{g,h}^{\vep}\|_{l^2}^2+\frac{\delta}{2}\|\delta_x^+(\brho_{h}-\brho_{g,h}^{\vep})\|_{l^2}^2\le E_h^{\vep}(\brho_{h})-E_h^{\vep}(\brho_{g,h}^{\vep}).
\ee
\end{theorem}
\begin{proof}
We follow a similar procedure as in the proof of Theorem \ref{thm:conv_rate}.
Define
\be
f(t)=E_h^{\vep}(\brho_{g,h}^{\vep}+t(\brho_{h}-\brho_{g,h}^{\vep})),\quad \text{ where }t\in[0,1].
\ee
Noticing that $\brho_{g,h}^{\vep}$ minimizes $E_h^{\vep}(\cdot)$ in $W_h$ \eqref{def:feasible},  we have
\be\label{proof:conv_ineq1}
f'(0)=(\mathbf{g}_h^{\vep}(\brho_{g,h}^{\vep}))^T(\brho_{h}-\brho_{g,h}^{\vep})\ge0,
\ee
where $\mathbf{g}_h^{\vep}$ is defined in \eqref{Eh_rho_gradient_matrix}.
Substituting \eqref{Eh_rho_gradient_matrix} into \eqref{proof:conv_ineq1},  we  get
\begin{align}\label{proof:conv_ineq2}
\mathbf{v}^T(\brho_{h}-\brho_{g,h}^{\vep})+R(\brho_{h},\brho_{g,h}^{\vep})\ge-\left[\nabla E^{\rm{kin}}(\brho_{g,h}^{\vep})\right]^T(\brho_{h}-\brho_{g,h}^{\vep}),
\end{align}
where
\be
R(\brho_{h},\brho_{g,h}^{\vep})=\left[\beta \brho_{g,h}^{\vep}+2\delta A\brho_{g,h}^{\vep}\right]^T(\brho_{h}-\brho_{g,h}^{\vep}).
\ee

On the other hand, a direct computation shows that
\be\label{proof:conv_eq1}
E_h^{\vep}(\brho_{h})-E_h^{\vep}(\brho_{g,h}^{\vep})
=h \left[ E^{\rm{kin}}(\brho_{h})-E^{\rm{kin}}(\brho_{g,h}^{\vep})\right]+h\mathbf{v}^T(\brho_{h}-\brho_{g,h}^{\vep})+I_1,
\ee
where  the difference of the interaction energy is contained in the term $I_1$, which can be explicitly computed as
\be
I_1=\frac{h\beta}{2}\left[(\brho_{h})^T\brho_{h}-(\brho_{g,h}^{\vep})^T\brho_{g,h}^{\vep}\right]+h\delta\left[ (\brho_{h})^TA\brho_{h}- (\brho_{g,h}^{\vep})^TA\brho_{g,h}^{\vep}\right].
\ee
Noticing that $h\brho_{h}^TA\brho_{h}=\|\delta_x^+\brho_{h}\|_{l^2}^2$ holds true for any $\brho_{h}\in W_h$ and $A^T=A$, we have
\begin{align}
I_1&=h\left(\frac{\beta}{2}(\brho_{h}+\brho_{g,h}^{\vep})+\delta A(\brho_{h}+\brho_{g,h}^{\vep})\right)^T(\brho_{h}-\brho_{g,h}^{\vep})\nonumber\\
&=hR(\brho_{h},\brho_{g,h}^{\vep})+D(\brho_{h},\brho_{g,h}^{\vep}),
\end{align}
where $D(\brho_{h},\brho_{g,h}^{\vep})=\frac{\beta}{2}\|\brho_{h}-\brho_{g,h}^{\vep}\|_{l^2}^2+\frac{\delta}{2}\|\delta_x^+(\brho_{h}-\brho_{g,h}^{\vep})\|_{l^2}^2$ and, therefore,
\begin{align}
E_h^{\vep}(\brho_{h})-E_h^{\vep}(\brho_{g,h}^{\vep})
&=h \left[E^{\rm{kin}}(\brho_{h})-E^{\rm{kin}}(\brho_{g,h}^{\vep})\right]\nonumber\\
&\quad+h[\mathbf{v}^T(\brho_{h}-\brho_{g,h}^{\vep})+R(\brho_{h},\brho_{g,h}^{\vep})]+D(\brho_{h},\brho_{g,h}^{\vep})\nonumber\\
&\ge h \left[E^{\rm{kin}}(\brho_{h})-E^{\rm{kin}}(\brho_{g,h}^{\vep})\right]\nonumber\\
&\quad-h\left[\nabla E^{\rm{kin}}(\brho_{g,h}^{\vep})\right]^T(\brho_{h}-\brho_{g,h}^{\vep})+D(\brho_{h},\brho_{g,h}^{\vep})\nonumber\\
&\ge D(\brho_{h},\brho_{g,h}^{\vep}),\label{proof:conv_ineq4}
\end{align}
where the first inequality is due to 
\eqref{proof:conv_ineq2} and the second inequality is due to the convexity of $E^{\rm{kin}}(\cdot)$ in $W_h$ \eqref{def:feasible}, which is proved in Lemma \ref{Eh_convex}.
The above inequality is exactly what we aim to show in Theorem \ref{thm:conv_h_E} noticing the definition of $D(\brho_{h},\brho_{g,h}^{\vep})$.
\end{proof}

Theorem \ref{thm:conv_h_E} offers a way to  measure the difference between a general state and the ground state  by considering the difference of the corresponding energies when $\beta>0$ and $\delta>0$.
In particular, we can study the discrete  error estimates based on Theorem \ref{thm:conv_h_E}. For simplicity, we use the notation $\Pi_h\rho_{g}^{\vep}$ to be the interpolation of the ground state $\rho_g^\vep$ \eqref{ground:reg} on the grid points, i.e.
\be\label{notation:interpolation:rho_g}
\Pi_h\rho_{g}^{\vep}=(\rho_g^\vep(x_1),\rho_g^\vep(x_2),\dots,\rho_g^\vep(x_{N-1}))^T \in \mathbb{R}^{N-1},
\ee
then we have the following result measuring the difference $\Pi_h\rho_{g}^{\vep}-\brho_{g,h}^{\vep}$.
\begin{theorem}\label{thm:conv_h}
Fix $\vep>0$ and consider $\mathbf{e}_h^{\vep}=\Pi_h\rho_{g}^{\vep}-\brho_{g,h}^{\vep}$. If $\beta>0$ and $\delta>0$, then we have
\be
\|\delta_+\mathbf{e}_h^{\vep}\|_{l^2}\lesssim h, \quad \|\mathbf{e}_h^{\vep}\|_{l^2}\lesssim h.
\ee
\end{theorem}
\begin{proof}
Denote $C_h$ to be the scaling constant such that $C_h\|\Pi_h\rho_{g}^{\vep}\|_{l^1}=1$. Then $C_h\Pi_h\rho_{g}^{\vep}\in W_h$ \eqref{def:feasible}.
Noticing that $\|\Pi_h\rho_{g}^{\vep}\|_{l^1}$ approximates $\|\rho_{g}^{\vep}\|_1=1$  via the composite trapezoidal rule,
 we have that
\be\label{proof:C_h1}
\|\Pi_h\rho_{g}^{\vep}\|_{l^1}-1=\|\Pi_h\rho_{g}^{\vep}\|_{l^1}-\int_a^b\,\rho_g^\vep\,dx=\frac{b-a}{12}\p_{x}^2\rho_g^\vep(\xi)h^2,
\ee
for some $\xi\in(a,b)$.
As a result,
\be\label{proof:C_h}
\left|C_h-1\right|\le C_1h^2,
\ee
for some $C_1$ depending on $\|\p_{x}^2\rho_g^\vep\|_{\infty}$, when $h$ is sufficiently small.

Notice that $\frac{1}{4}\frac{|\delta_x^+\rho_j|^2}{\rho_j+\rho_{j+1}+2\vep}$ and $\delta_x^+\rho_{j}$ are the second order accurate approximations of $\frac12|\p_x\sqrt{\rho_g^\vep+\vep}|^2$ and $\p_x\rho_g^\vep$ at $x=x_{j+\frac12}$, respectively. 
It is then obvious that $E_h^{\vep}(\cdot)$ approximates $E^{\vep}(\cdot)$ by combing two second order accurate quadratures, i.e.
the composite midpoint rule and the composite trapezoidal rule.
Following a similar procedure, we get
 \be
 |E^{\vep}(\rho_g^{\vep})-E_h^{\vep}(C_h\Pi_h\rho_{g}^{\vep})|\le C_2h^2,
 \ee
 where the constant $C_2$ depends on $\|\p_x^l\rho_g^{\vep}\|_{\infty}$ for $l=0,1,2,3$.

Besides, if we consider the positive piecewise linear  function $\trho_{g}^{\vep}(x)$ such that
$\trho_{g}^{\vep}(x_j)=\rho_{g,j}^{\vep}$,
then $\trho_{g}^{\vep}\in W$ \eqref{def:feasible_rho} and $\frac{\delta}{2}|\p_x\trho_g^\vep|^2$ is piecewise constant.
A similar, but more careful analysis shows that
\be
E^{\vep}(\trho_{g}^{\vep})\le E_h^{\vep}(\brho_{g,h}^{\vep})+C_3h^2,
\ee
where the constant $C_3$ depends on the uniform bound of $\|\delta_x^+\brho_{g,h}^\vep\|_{l^2}$
with respect to $h$.
The bound can be expressed via $E^{\vep}(\rho_g^{\vep})$.
The details are omitted here for brevity.

Recalling Theorem \ref{thm:conv_h_E}, we have
\begin{align}
&\frac{\beta}{2}\|C_h\Pi_h\rho_{g}^{\vep}-\brho_{g,h}^{\vep}\|_{l^2}^2+\frac{\delta}{2}\|\delta_+(C_h\Pi_h\rho_{g}^{\vep}-\brho_{g,h}^{\vep})\|_{l^2}^2
\le E_h^{\vep}(C_h\Pi_h\rho_{g}^{\vep})-E_h^{\vep}(\brho_{g,h}^{\vep})\nonumber\\
&\le E^{\vep}(\rho_g^{\vep})-E^{\vep}(\trho_g^{\vep})+(C_2+C_3)h^2
\le (C_2+C_3)h^2.
\end{align}
The fact that $E^{\vep}(\rho_g^{\vep})\le E^{\vep}(\trho_g^{\vep})$ is applied in the last step.
For $\beta>0$ and $\delta>0$, it is then obvious that
\be
\|C_h\Pi_h\rho_{g}^{\vep}-\brho_{g,h}^{\vep}\|_{l^2}\lesssim h \text{ and } \|\delta_+(C_h\Pi_h\rho_{g}^{\vep}-\brho_{g,h}^{\vep})\|_{l^2}\lesssim h,
\ee
which further implies that
\be
\|\mathbf{e}_h^{\vep}\|_{l^2}\le\|C_h\Pi_h\rho_{g}^{\vep}-\brho_{g,h}^{\vep}\|_{l^2}+|C_h-1| \|\Pi_h\rho_{g}^{\vep}\|_{l^2} \lesssim h,
\ee
and
\be
\|\delta_+\mathbf{e}_h^{\vep}\|_{l^2}\le\|\delta_+(C_h\Pi_h\rho_{g}^{\vep}-\brho_{g,h}^{\vep})\|_{l^2}+|C_h-1| \|\delta_+(\Pi_h\rho_{g}^{\vep})\|_{l^2} \lesssim h.
\ee
\end{proof}
\begin{remark}\label{rem:L2-order}
If we further assume that
$|\rho_g^\vep|_{H^2}$ is bounded, we can expect $\|\mathbf{e}_h^{\vep}\|_{l^2}\lesssim h^2$ by the Bramble-Hilbert lemma and a standard scaling argument.
\end{remark}

Finally, we study the convergence of $\brho_{g,h}^{\vep}$ as $\vep\to0^+$ with a fixed $h$ and the result is summarized as follows.
\begin{theorem}\label{thm:conv_h_E_discrete}
When $\beta\ge0$ and $\delta\ge0$, we have that
\be\label{result1}
\lim_{\vep\to0^+}\brho_{g,h}^{\vep}=\brho_{g,h}^{0},
\ee
where $\brho_{g,h}^{0}$ is the ground state of $E_h^{\vep}(\cdot)$ \eqref{Eh_rho} with $\vep=0$.
\end{theorem}
\begin{proof}

The compactness of the finite dimensional feasible set $W_h$ \eqref{def:feasible} implies that, for any sequence $\vep_n\to0$,  there exists a subsequence, denoted as the sequence itself for simplicity, such that $\lim_{n\to\infty}\brho_{g,h}^{\vep_{n}}=\tbrho_{h}$ for some $\tbrho_{h}\in W_h$. Consequently, $\lim_{n\to\infty}E^{\vep_{n}}_{h}(\brho_{g,h}^{\vep_{n}})= E^{0}_{h}(\tbrho_{h})$.

On the other hand, for any $\vep$, we have $E_h^{\vep}(\brho_{g,h}^{\vep})\le E_h^{\vep}(\brho_{g,h}^{0})\le E_{h}^{0}(\brho_{g,h}^{0})$, which indicates that $E^{0}_{h}(\tbrho_{h})\le E^{0}_{h}(\brho_{g,h}^{0})$. Therefore, we must have $\tbrho_{h}=\brho_{g,h}^{0}$ due to the uniqueness of the ground state. Noticing that the above argument holds true for any sequence $\vep_n\to0$, we get the conclusion \eqref{result1}.
\end{proof}

The main works in Section \ref{sec:formulation} and Section \ref{sec:scheme} can be summarized in the following diagram.
Since  it is difficult to study the convergence of $\brho_{g,h}^{0}$ to $\rho_g$ directly, as indicated by the dash arrow in the diagram,
we study an approximating problem, i.e. the convergence of $\brho_{g,h}^{\vep}$ to $\rho_g$, which is indicated by the double arrow in the diagram.
As shown in the diagram, we have studied it in two steps,
namely the convergence of $\brho_{g,h}^{\vep}$ to $\rho_g^\vep$ as $h\to0^+$ (or $N\to\infty$), which has been proved in Theorem \ref{thm:conv_h},
and the convergence of $\rho_g^{\vep}$ to $\rho_g$ as $\vep\to0^+$, which has been studied in Theorem \ref{thm:gm_conv} and Theorem \ref{thm:conv_rate}.
Though not related to our main problem, the convergence of $\brho_{g,h}^{\vep}$ to $\brho_{g,h}^{0}$ has been studied as well in Theorem \ref{thm:conv_h_E_discrete}.

\begin{center}
\begin{tikzpicture}[font=\small\sffamily\bfseries,very thick]
  \matrix (m) [matrix of math nodes,row sep=6em,column sep=15em,minimum width=2em]
  {
     \brho_{g,h}^{\vep} & \brho_{g,h}^{0} \\
     \rho_g^{\vep} & \rho_g \\};
  \path[-stealth]
    (m-1-1)
    	    edge node [left,align=right,pos=.5] {$h\to0^+$} (m-2-1)
            edge node [above] {$\vep\to0^+$} (m-1-2)
            edge [double]  (m-2-2)
    (m-2-1) edge node [above] {$\vep\to0^+$} (m-2-2)
    (m-1-2) edge [dashed] (m-2-2);
\end{tikzpicture}
\end{center}

\section{An accelerated projected gradient method}\label{sec:APG}
Numerous optimization techniques have been developed for solving a constrained convex  optimization problem like \eqref{problem:E_rho_FD}.
Among them, the accelerated projected gradient method (APG) has been widely used due to its efficiency and easy implementation.
For completeness, a brief introduction of the method is reviewed in this section.

The APG method is a special case of the accelerated proximal method,
which is a generic method proposed in \cite{Beck,Nesterov} to solve problems of the type
\be\label{problem:apg}
\min_{\mathbf{u}\in\mathbb{R}^n} F(\mathbf{u}):=f(\mathbf{u})+g(\mathbf{u}),
\ee
with $f(\mathbf{u})$ being convex and of type $C^{1,1}$,  i.e. continuously differentiable with Lipschitz continuous gradient,
and $g(\mathbf{u})$ being convex, `simple', but possibly non-smooth.
`Simple' means easy computation of the proximal operation of $g$, i.e.
\be
\textrm{prox}_g(\mathbf{w}):=
\argmin_{\mathbf{u}\in\mathbb{R}^n}\left(g(\mathbf{u})+
\frac{1}{2}\|\mathbf{u}-\mathbf{w}\|_{l^2}^2\right),
\ee
for any given $\mathbf{w}\in\mathbb{R}^n$.
When $g(\mathbf{u})=\mathbb{I}_{Q}(\mathbf{u})$, where $\mathbb{I}_{Q}(\cdot)$ is the indicator function defined as
\begin{align}
\mathbb{I}_{Q}(\mathbf{u})=\left\{\begin{array}{lll}
0, &\text{ if }\mathbf{u}\in Q,\\
\infty, &\text{ otherwise,}\\
\end{array}\right.
\end{align}
for some  convex set $Q$,
the accelerated proximal method becomes the APG method  since now the operator $\textrm{prox}_g(\mathbf{w})$ becomes  the $l^2$-projection of $\mathbf{w}$ onto $Q$, i.e.
\be
\textrm{prox}_g(\mathbf{w})=\textrm{proj}_{Q}(\mathbf{w}):=\argmin_{\mathbf{u}\in Q}\|\mathbf{u}-\mathbf{w}\|_{l^2}^2.
\ee

The problem \eqref{problem:E_rho_FD} can be reformulated easily into  the form \eqref{problem:apg} by taking
\be
f(\brho_h)=E_h^{\vep}(\brho_h), \quad g(\brho_h)=\mathbb{I}_{W_h}(\brho_h),
\ee
where $\brho_h\in \mathbb{R}^{N-1}$ and $W_h\subset\mathbb{R}^{N-1}$ is the feasible set defined as \eqref{def:feasible}.
However, a further modification of  $E_h^{\vep}(\cdot)$ is necessary as follows,
\be\label{Eh_rho_APG}
\tE_h^{\vep}(\brho_{h})=h\sum_{j=0}^{N-1}\left[\frac{1}{4}
\frac{|\delta_x^+\rho_j|^2}{|\rho_j|+|\rho_{j+1}|+2\vep}+V(x_j)
|\rho_j|+\frac{\beta}{2}\rho_j^2+\frac{\delta}{2}\left|\delta_x^+
\rho_{j}\right|^2\right].
\ee
 Here we impose the absolute sign on $\rho_j$ to extend the domain of $E_h^{\vep}(\cdot)$ from the feasible set $W_h$ to the whole space $\mathbb{R}^{N-1}$ in order to apply the APG method.
This modification is necessary to avoid possible breakdown of the optimization process.
 Denote
\be
\tilde{\mathbf{g}}_h^{\vep}(\brho_{h})=\left( \frac{\p \tE_h^{\vep}}{\p\rho_1},\frac{\p \tE_h^{\vep}}{\p\rho_2},\dots,\frac{\p \tE_h^{\vep}}{\p\rho_{N-1}}\right)^T.
\ee
Now its $j$-th component  becomes
\be\label{Eh_rho_gradient_APG}
\tilde{\mathbf{g}}_h^{\vep}[j]:=\frac{\partial \tE_h^{\vep}}{\partial \rho_j}=h\left[-\frac{\delta_x^+\tf_{j-1}}{2}-\frac{\tf_{j-1}^2+\tf_j^2}{4}s_j+V(x_j)s_j+\beta \rho_j-\delta(\delta_x^2\rho_j)\right],
\ee
where
\be\label{def:s_f_APG}
s_j=\textrm{sign}(\rho_j),\quad \tf_j=\tf_j(\brho_{h}):=\frac{\delta_x^+\rho_j}{|\rho_j|+|\rho_{j+1}|+2\vep},\quad \text{ for }  j=0,1,\dots, N-1.
\ee
Obviously,  for any $\brho_h\in W_h$ \eqref{def:feasible},  we have
\be
\tE_h^{\vep}(\brho_h) = E_h^{\vep}(\brho_h),\quad \tilde{\mathbf{g}}_h^{\vep}(\brho_h)=\mathbf{g}_h^{\vep}(\brho_h),
\ee
and the minimizers of $\tE_h^{\vep}(\cdot)$ in $\mathbb{R}^{N-1}$ lie in the feasible set $W_h$ \eqref{def:feasible} or $-W_h$.


In practice, the fact that $g(\mathbf{u})$ is `simple' is crucial for the efficiency of the APG method since the evaluation of the proximal operator of $g(\cdot)$ is done at each iteration and makes up the most time-consuming part in optimization process.
Therefore, it is necessary to check whether the function $\mathbb{I}_{W_h}(\brho_h)$ is `simple' or not.
Luckily, the feasible set $W_h$ is a simplex and the projection onto $W_h$ can be efficiently computed with an average computational cost $O(N\log N)$ or even $O(N)$ \cite{Brucker,Duchi,Pardalos}.

Now we are ready to show the detailed scheme for  optimizing  \eqref{Eh_rho_APG}.
The general framework comes from FISTA, a simple and popular APG method proposed in \cite{Beck}.
For simplicity, we introduce the following two notations, $Q_L(\mathbf{u},\brho_h)$ and $p_L(\brho_h)$.
For given $L > 0$, $\mathbf{u}\in\mathbb{R}^{N-1}$ and $\brho_h\in\mathbb{R}^{N-1}$, we denote
\be
Q_L(\mathbf{u},\brho_h)=\tE_h^{\vep}(\brho_h)+(\mathbf{u}-\brho_h)^T
\tilde{\mathbf{g}}_h^{\vep}(\brho_h)+\frac{L}{2}\|\mathbf{u}-\brho_h\|_2^2 +\mathbb{I}_{W_h}(\mathbf{u}),
\ee
and
\be
p_L(\brho_h) := \argmin_{\mathbf{u}\in\mathbb{R}^{N-1}}Q_L(\mathbf{u}, \brho_h).
\ee
Obviously, $Q_L(\cdot,\brho_h)$ is a quadratic approximation of $\tE_h^{\vep}(\cdot)+\mathbb{I}_{W_h}(\cdot)$ near  $\brho_h$, and $p_L(\brho_h)$, the minimizer of $Q_L(\cdot,\brho_h)$, exists and is unique for any given $\brho_h$.
Furthermore, a simple algebra shows that
\begin{align}
p_L(\brho_h)&=\argmin_{\mathbf{u}\in W_h} \left\|\mathbf{u}- \left(\brho_h-\frac{1}{L}\tilde{\mathbf{g}}_h^{\vep}(\brho_h)\right)
\right\|_{l^2}^2\nonumber\\
&=\textrm{proj}_{W_h}\left(\brho_h-\frac{1}{L}\tilde{\mathbf{g}}_h^{\vep}(\brho_h)\right),
\end{align}
i.e. $p_L(\brho_h)$ is the projection of $\brho_h-\frac{1}{L}\tilde{\mathbf{g}}_h^{\vep}(\brho_h)$ onto the feasible set $W_h$ \eqref{def:feasible}.
With all these preparations, we are  ready to show the detailed algorithm  in  Algorithm \ref{alg:update_APG}.
As shown in \cite{Beck}, we have $E_h^{\vep}(\brho_h^{(k)})-E_h^{\vep}(\brho_h^{(k+1)})\lesssim\mathcal{O}(1/k^2)$, where $\brho_h^{(k)}$ is the vector computed after $k$ iterations.
In other words, a quadratic convergence rate in energy  could be expected.

We remark here  that there have been many efficient Matlab packages developed for the problem \eqref{problem:apg} based on APG, such as the TFOCS \cite{TFOCS}.
These packages, which are based on similar ideas as FISTA,  are  usually more efficient and robust, though more complicated, than FISTA.
Therefore, these packages could be used to replace the FISTA in the  Algorithm \ref{alg:update_APG} for a better numerical performance.

\begin{algorithm}
\caption{rDF-APG: Optimize \eqref{Eh_rho_APG} via FISTA}\label{alg:update_APG}
\begin{algorithmic}[1]
\STATE Choose a sufficiently small $\vep>0$ and a proper initial guess $\brho_h^{(0)}$.
\STATE Choose $L_0>0$, $\eta>1$.
\STATE $k\gets1$, $t_1\gets1$,  $\mathbf{y}_1\gets \brho_h^{(0)}$, $\bar{L}\gets L_0$
\STATE $\tilde{\brho_h}\gets p_{\bar{L}}(\mathbf{y}_1)$
\WHILE{$k=1$ or $\|\brho_h^{(k-1)}-\brho_h^{(k-2)}\|_{l^2}>\rm{tolerence}$ }
\WHILE{$E_h^{\vep}(\tilde{\brho_h})>Q_{\bar{L}}(\tilde{\brho_h},\mathbf{y}_k)$}
\STATE $\bar{L}\gets \eta\bar{L}$, 
 $\tbrho_h\gets p_{\bar{L}}(\mathbf{y}_k)$
\ENDWHILE
\STATE $\brho_h^{(k)}\gets\tbrho_h$, $t_{k+1}\gets\frac{1+\sqrt{1+4t_k^2}}{2}$, $L_{k}\gets\bar{L}$
\STATE $\mathbf{y}_{k+1}\gets \brho_h^{(k)}+\left(\frac{t_k-1}{t_{k+1}}\right)(\brho_h^{(k)}-\brho_h^{(k-1)})$
\STATE $k\gets k+1$
\ENDWHILE
\STATE $\brho_{g,h}^{\vep}\gets \brho_h^{(k-1)}$
\end{algorithmic}
\end{algorithm}

\begin{remark}
The approximate ground states obtained in
\cite{Ruan} for different parameter regimes can be
taken as good initial guess in the rDF-APG method.
\end{remark}
\begin{remark}
The two-grid technique, where we start from a coarse mesh and make the refined solution on the coarse mesh to be the initial guess for the fine mesh, will be useful to speed up the algorithm.
Similarly, when $\vep$ is small, it is useful to apply the continuation technique, which means we start with a relatively large $\vep$ and use the result as the initial guess for a smaller $\vep$.
\end{remark}
\begin{remark}
Based on our numerical experiments not reported here for brevity, we observed that  when $\vep$ is extremely small,
it is better to choose $\eta$ in Algorithm \ref{alg:update_APG} close to 1. 
\end{remark}
\begin{remark}
Numerical experiments suggest that a further regularization  could lead to a much better numerical performance.
With the regularization, we get the following new regularized energy functional
\be\label{energy_regularizeV}
\hat{E}^{\vep}(\rho)=\int_{\mathbb{R}^d}\left[\frac{|\nabla\rho|^2}
{8(|\rho|+\vep)}+V(\bx)\left(\sqrt{\rho^2+\vep^2}-\vep\right)+\frac{\beta}{2}
\rho^2+\frac{\delta}{2}|\nabla\rho|^2\right]d\bx,
\ee
where $\rho\in W$ \eqref{def:feasible_rho}.
Obviously, $\hat{E}^{\vep}(\cdot)$ is convex in $W$ \eqref{def:feasible_rho}.
Denote $\hat{\rho}_g^{\vep}$ to be the ground state of $\hat{E}^{\vep}(\cdot)$, which is defined as
\be\label{ground:regV}
\hat{\rho}_g^{\vep}=\argmin_{\rho\in W} \hat{E}^{\vep}(\rho).
\ee
It can be shown that Theorem \ref{thm:mres}, Theorem \ref{thm:gm_conv} and Theorem \ref{thm:conv_rate} still hold true if we replace $\rho_{g}^{\vep}$ by $\hat{\rho}_g^{\vep}$ in the theorems.
A detailed description of its finite difference discretization can be found in Appendix \ref{appendix:reg_V_FD}.
\end{remark}

\section{Numerical Results}\label{sec:numeric}
In this section, we illustrate the performance of the rDF-APG method proposed in this paper.
In particular, we will verify the spatial accuracy result in Theorem \ref{thm:conv_h} and test the convergence of $\rho_g^\vep$ as $\vep\to0^+$.
Besides,  we compare the rDF-APG method with the regularized Newton method \cite{BaoWuWen} to show the great advantage of our method when the HOI effect is dominant.
Applications of our method to multi-dimensional problems are also reported.
As a remark, all the problems in this section are formulated via $\hat{E}^{\vep}(\cdot)$ \eqref{energy_regularizeV} and all numerical ground state densities are  denoted as $\brho_{g,h}^{\vep}$ for simplicity.

\subsection{Spatial accuracy}
To check the spatial accuracy of the ground state $\brho_{g,h}^{\vep}$ defined in \eqref{problem:E_rho_FD}, we fix $\vep>0$.
For simplicity, only the following 1D case  in Example \ref{example:accuracy_test} is considered.
For multi-dimensional problems, the results are similar.
\begin{myexample}\label{example:accuracy_test}
Consider the external potential to be $V(x)=x^2/2$ and the computational domain $\Omega=(-M,M)$ with $M$ sufficiently large. 
The two nonnegative parameters $\beta$ and $\delta$ can be chosen arbitrarily.
In particular, we  consider the following two cases with $M=16$.
\begin{itemize}
\item Case I: $\,\,\beta=10$ and $\delta=0$.
\item Case II: $\beta=10$ and $\delta=10$.
\end{itemize}
Obviously, there is no HOI effect in Case I.
\end{myexample}
For the fixed $\vep=10^{-3}$, the `exact' ground state density $\rho_{g}^\vep$ for Case I and Case II in Example \ref{example:accuracy_test} can be accurately approximated via the rDF-APG method with an extremely small mesh size, say $h=\frac{1}{4096}$, and an extremely small tolerance, say $10^{-12}$.
If we denote the corresponding energy as $E_{\rm{ex}}^\vep:=E^{\vep}(\rho_{g}^\vep)$,
it can be computed that $E_{\rm{ex}}^\vep\approx1.927$ in Case I and  $E_{\rm{ex}}^\vep\approx2.207$ in  Case II.

Table \ref{tab:error_dh_test1} and Table \ref{tab:error_dh_test2} show the spatial  accuracy of $\brho_{g,h}^{\vep}$ for Case I and Case II in Example \ref{example:accuracy_test}, respectively.
From the tables, we can observe the second order accuracy  of energy and first order accuracy of $\delta_+\brho_{g,h}^{\vep}$ in $l^2$-norm, which are consistent with Theorem \ref{thm:conv_h}.
In fact, Table \ref{tab:error_dh_test1} and  \ref{tab:error_dh_test2}  show something more than expected.
Firstly, we observe a roughly  second order accuracy  in $l^2$-norm of $\brho_{g,h}^{\vep}$ while only the first order accuracy can be proved in Theorem \ref{thm:conv_h}.
 One possible explanation is that the ground state is regular and thus we can get a higher order accuracy noticing Remark \ref{rem:L2-order}.
Besides, it seems that  we still have the first order accuracy of $\delta_+\brho_{g,h}^{\vep}$ in $l^2$-norm  when $\delta=0$,
although we can no longer prove such a result following a similar procedure as in Theorem \ref{thm:conv_h}.
\begin{table}
  \centering
  \begin{tabular}{c c c c c c c }
    \hline
    Error & $h=1/2$ & $h/2$ & $h/2^2$ & $h/2^3$ & $h/2^4$ & $h/2^5$ \\
    \hline
    $|E_{g,h}^{\vep}-E_{\rm{ex}}^\vep|$  & 7.74E-3 & 1.97E-3 & 4.93E-4 & 1.23E-4 & 3.06E-5 & 7.59E-6 \\
    rate & -& \textbf{1.98} & \textbf{2.00} & \textbf{2.00} & \textbf{2.01} & \textbf{2.01} \\
    \hline
    $\|\mathbf{e}_h^{\vep}\|_{l^2}$  & 1.17E-1 & 2.97E-2 & 7.45E-3 & 1.84E-3 & 4.48E-4 & 8.74E-5 \\
    rate & - & 1.98 & 2.00 & 2.01 & 2.04 & 2.36 \\
    \hline
    $\|\delta_+\mathbf{e}_h^{\vep}\|_{l^2}$  & 2.18E-3 & 1.07E-3 & 5.33E-4  & 2.66E-4 & 1.33E-4 & 6.65E-5 \\
    rate & -  & \textbf{1.02} & \textbf{1.01}  & \textbf{1.00} & \textbf{1.00} & \textbf{1.00} \\
    \hline
     $\|\mathbf{e}_h^{\vep}\|_{l_{\infty}}$  & 5.12E-3 & 1.29E-3 & 3.27E-4 & 8.10E-5 & 1.98E-5 & 3.73E-6 \\
    rate & -  & 1.99 & 1.98 & 2.01 & 2.03 & 2.41 \\
    \hline
  \end{tabular}
  \caption{Spatial accuracy of rDF-APG for Case I in Example \ref{example:accuracy_test}. In the table, we denote $\mathbf{e}_h^{\vep}:=\Pi_h\rho_{g}^{\vep}-\brho_{g,h}^{\vep}$,  where $\Pi_h\rho_{g}^{\vep}$ is defined in \eqref{notation:interpolation:rho_g}.}
  \label{tab:error_dh_test1}
\end{table}

\begin{table}
  \centering
  \begin{tabular}{c c c c c c c }
    \hline
    Error & $h=1/2$ & $h/2$ & $h/2^2$ & $h/2^3$ & $h/2^4$ & $h/2^5$ \\
    \hline
    $|E_{g,h}^{\vep}-E_{\rm{ex}}^\vep|$   & 9.07E-3 & 2.28E-3  & 5.73E-4 & 1.43E-4 & 3.56E-5 & 8.81E-6 \\
    rate & - & \textbf{1.99} & \textbf{1.99} & \textbf{2.00} & \textbf{2.01} & \textbf{2.01} \\
    \hline
    $\|\mathbf{e}_h^{\vep}\|_{l^2}$  & 5.98E-2 & 1.44E-2 & 3.59E-3 & 8.88E-4 & 2.19E-4 & 4.91E-5 \\
    rate & - & 2.06 & 2.00 & 2.02 & 2.02 & 2.16 \\
    \hline
    $\|\delta_+\mathbf{e}_h^{\vep}\|_{l^2}$   & 1.43E-3 & 7.10E-4 & 3.54E-4 & 1.77E-4 & 8.85E-5 & 4.42E-5 \\
    rate & -  & \textbf{1.01} & \textbf{1.00}  & \textbf{1.00} & \textbf{1.00} & \textbf{1.00} \\
    \hline
     $\|\mathbf{e}_h^{\vep}\|_{l_{\infty}}$  & 2.33E-3 & 6.01E-4 & 1.51E-4 & 3.77E-5 & 9.35E-6 & 2.14E-6 \\
    rate & -  & 1.95 & 2.00  & 2.00 & 2.01 & 2.13 \\
    \hline
  \end{tabular}
  \caption{Spatial accuracy of rDF-APG for Case II in Example \ref{example:accuracy_test}. In the table, we denote $\mathbf{e}_h^{\vep}:=\Pi_h\rho_{g}^{\vep}-\brho_{g,h}^{\vep}$,  where $\Pi_h\rho_{g}^{\vep}$ is defined in \eqref{notation:interpolation:rho_g}.}
  \label{tab:error_dh_test2}
\end{table}


\subsection{Convergence of ground states}
To verify numerically the convergence of the ground state densities of the regularized energy functionals,
we choose a sequence $\{\vep_n\}$ such that $\lim_{n\to\infty}\vep_n=0$.
For each $\vep_n$, we choose a sufficiently small mesh size  to make the spatial error  negligible.
In general, the smaller $\vep$ is, the finer mesh is needed.

Here we consider the 1D example as shown in Case II of Example \ref{example:accuracy_test}.
Denote $\rho_{g}^\vep$ to be the `exact' ground state density for the fixed $\vep>0$, which is numerically computed via the rDF-APG method,
and $\rho_{g}$ to be the `exact' numerical  ground state density of the original energy functional \eqref{energy_rho}, i.e. when $\vep=0$, which is approximated by $\rho_{g}^\vep$ with a sufficiently small $\vep$, say $\vep=10^{-8}$.
Figure \ref{fig:sol_vep} shows the ground state densities computed for different $\vep\ge0$.
It can be observed that $\rho_{g}^\vep$ is less concentrated than $\rho_{g}$, but will converge to  $\rho_{g}$  as $\vep\to0^+$.
More detailed comparisons are made in Table \ref{tab:error_vep_test2}.
From Table \ref{tab:error_vep_test2}, we find that $\rho_{g}^\vep$
and $E^\vep(\rho_{g}^\vep)$ converge to $\rho_{g}$
and $E(\rho_{g})$, respectively,  almost linearly with respect to the regularization parameter $0<\vep\ll1$.

\begin{figure}[tbhp]
\centering
\subfloat{\includegraphics[height=5cm,width=14cm,angle=0]{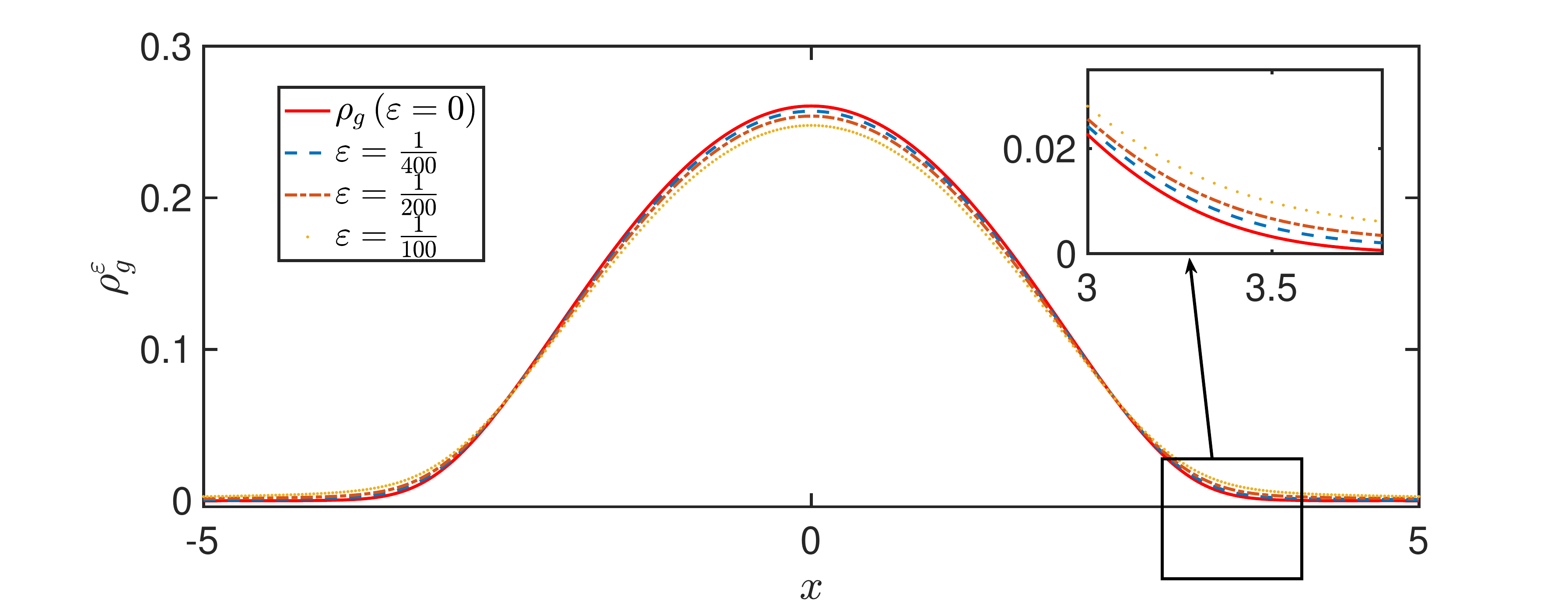}}
\caption{Comparison of $\rho_{g}^\vep$ \eqref{ground:reg} computed via the rDF-APG method and the `exact' ground state density $\rho_{g}$ \eqref{ground:ori}
for Case II in Example \ref{example:accuracy_test}. }
\label{fig:sol_vep}
\end{figure}


\begin{table}
  \centering
  \begin{tabular}{c c c c c c c c c}
    \hline
    $\vep$ & $10^{-1}$ & $10^{-2}$ & $10^{-3}$ & $10^{-4}$ & $10^{-5}$ & $10^{-6}$ \\
    \hline
    $|E^{\vep}(\rho_{g}^{\vep})-E(\rho_{g})|$ &  1.04E0 & 2.04E-1 & 2.84E-2 & 3.70E-3 & 4.56E-4 & 5.62E-5 \\
    rate & -& 0.71 & 0.86 & 0.88 & 0.90 & 0.92 \\
    \hline
    $\|\rho_{g}^{\vep}-\rho_{g}\|_2$ & 1.54E-1 & 2.45E-2 & 2.73E-3 & 2.95E-4 & 3.08E-5 & 3.08E-6  \\
    rate & -& 0.80 & 0.95 & 0.97 & 0.98 & 1.00 \\
    \hline
    $\|\rho_{g}^{\vep}-\rho_{g}\|_{\infty}$ & 8.64E-2 & 1.28E-2 & 1.43E-3 & 1.54E-4 & 1.52E-5 & 1.70E-6 \\
    rate & -  & 0.83  & 0.95 & 0.97 & 1.00 & 0.95 \\
    \hline
  \end{tabular}
  \caption{Convergence test of the ground state densities as $\vep\to0^+$ for Case II in Example \ref{example:accuracy_test}.}
  \label{tab:error_vep_test2}
\end{table}


The computational cost is highly related to the value of $\vep$.
For Case II of Example \ref{example:accuracy_test}, numerical experiments indicate that it would be difficult to make $\|\brho_h^{(k+1)}-\brho_h^{(k)}\|_{l^2}$ converge within a very small tolerance, say $10^{-10}$, when $\vep$ is extremely small, say $\vep\ll10^{-8}$,  where $\brho_h^{(k)}$ is the intermediate state computed after $k$ iterations.
The relatively slow convergence of the intermediate states $\brho_h^{(k)}$ when  $\vep$ is  extremely small is  the main drawback of our method.
A numerical evidence is shown in Table \ref{tab:CPU_cmp_vep}, from where we can clearly observe the increasing computational cost as $\vep\to0$.
In particular, the method fails when $\vep=0$.
On the other hand, though not explicitly shown here for brevity, numerical experiments indicate that our method works well if we are only concerned with the ground state energy, even when $\vep$ is very small, since we can choose a relatively larger tolerance to have a much faster convergence.

\begin{table}
  \centering
  \begin{tabular}{|c|c|c|c|c|c| }
	\hline
	$\vep$ & $10^{-2}$ & $10^{-4}$ & $10^{-6}$ & $10^{-8}$ & $0$ \\
		\hline
	CPU time &  5.7s & 18.3s & 60.7s & 159.4s & N.A.  \\
    \hline
  \end{tabular}
  \caption{Effect of $\vep$ on the efficiency of the rDF-APG method. In the numerical test, we choose  $\beta=10$, $\delta=100$, $h=\frac{1}{64}$, $\Omega=(-16,16)$ and a very small tolerance $10^{-10}$.
  The two-grid technique is applied to speed up.
  The result is accurate to the first digit.
  The method fails when $\vep=0$. }
  \label{tab:CPU_cmp_vep}
\end{table}

\subsection{Effect of interaction strength}
The contact interaction strength $\beta$ and the HOI strength $\delta$  affect the ground state in different ways.
As an example, we consider the problem defined in the whole space under the harmonic potential $V(x)=\frac{x^2}{2}$.
In Figure \ref{fig:ground_1D}, (a) and (b) show  the ground state densities with different choices of $\beta$ and $\delta$.
As shown in the figure, both the strong contact interaction and the strong HOI effect will spread the ground state, but in different ways.
The detailed limiting Thomas-Fermi approximations of the ground states could be referred to \cite{mgpe-th,Ruan}.
The subfigures (c) and (d)
in Figure \ref{fig:ground_1D} show explicitly how the value of $\beta$ and $\delta$ is related to the difference $\rho_g^\vep-\rho_g$.
Comparing with (a) and (b) in Figure \ref{fig:ground_1D}, when $\beta\gg1$, the difference will almost be a constant in the region where $\rho_g$ is obviously bigger than 0.
When $\delta\gg1$, the pattern of the difference will be different.
In fact, by considering the Euler-Lagrange equation satisfied by $\rho_g$ and $\rho_g^\vep$, we can show that the difference in this case is roughly a quadratic function in the region where $\rho_g$ is obviously bigger than 0.

The strength of interaction would affect the numerical performance  of the rDF-APG method as well.
We expect that the rDF-APG method works better when
 the interaction is stronger since the interaction energy terms are quadratic with respect to the density function formulation,  which is the ideal form for most optimization techniques.
As a result, we can expect a great advantage of the rDF-APG method over methods based on the wave function formulation when the interaction energy part, especially the HOI part, is dominant.
To better illustrate the advantage of the density function formulation we introduced,
we compare the rDF-APG method with the regularized Newton method \cite{Ruan_thesis,BaoWuWen}, which is one of the state-of-the-art numerical methods for computing ground states of BEC \cite{BaoWuWen}. In fact, both methods
share the similar strategy in designing the numerical methods, i.e.
first discretize the energy functional to get a finite-dimensional constrained optimization problem and then adapt modern optimization techniques to compute the ground state. The difference is that the regularized Newton method adopts the wave function formulation while
the rDF-APG method uses the density function formulation.
We choose all the setup of the problem, including the initial data and the stopping criteria, to be exactly the same for a fair comparison. 
The initial guess is chosen to be of Gaussian type and the stopping tolerance is chosen to be  $10^{-8}$.
All the other parameters in the regularized Newton method are chosen to be the default ones.
For the rDF-APG method, we fix $\vep=10^{-4}$ and choose $\eta=1/0.9$.
We run the two methods on the same laptop and the CPU time for different cases are shown in Table \ref{tab:CPU_cmp}.
Although the CPU time will not be exactly the same for each run, the results shown in the table are computed by averaging, which are accurate to the first digit and enough to verify our expectations.
As shown in the table,
the advantage of the rDF-APG method over the regularized Newton method is not obvious
when $\beta$ and $\delta$ are small.
However, when one of  $\beta$ and $\delta$ is large, especially when the $\delta$ term is dominant, there is a huge difference between the regularized Newton method and the rDF-APG method  considering computational cost.
What's more, the stronger the interaction is, the faster the rDF-APG method becomes, which is completely different from the regularized Newton method.
As a conclusion, the rDF-APG method works for all positive $\beta$ and $\delta$, and is extremely suitable to cases with strong interaction effect.


\begin{figure}[tbhp]
\centerline{
\psfig{
figure=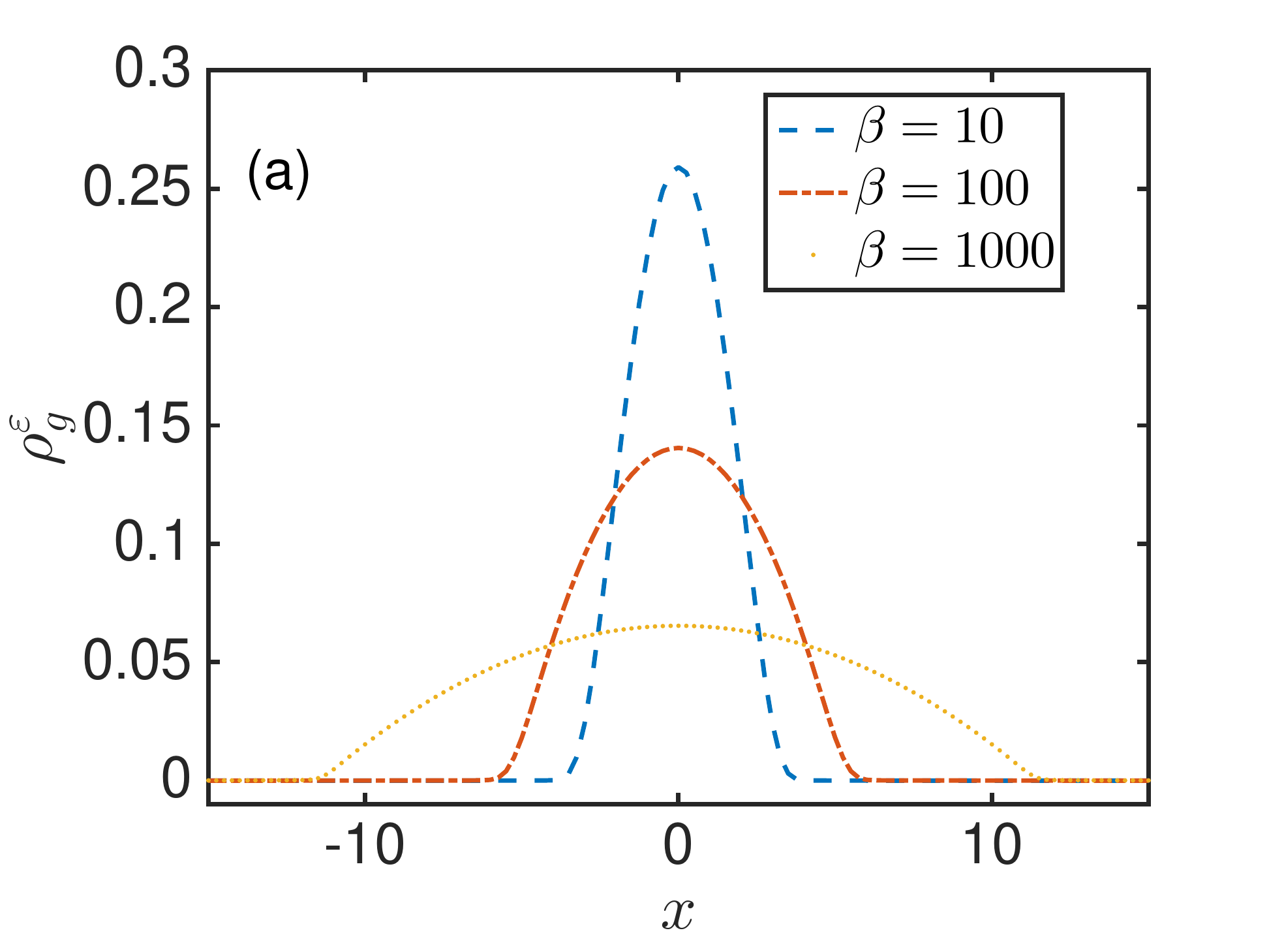,height=5cm,width=7cm,angle=0}
\psfig{
figure=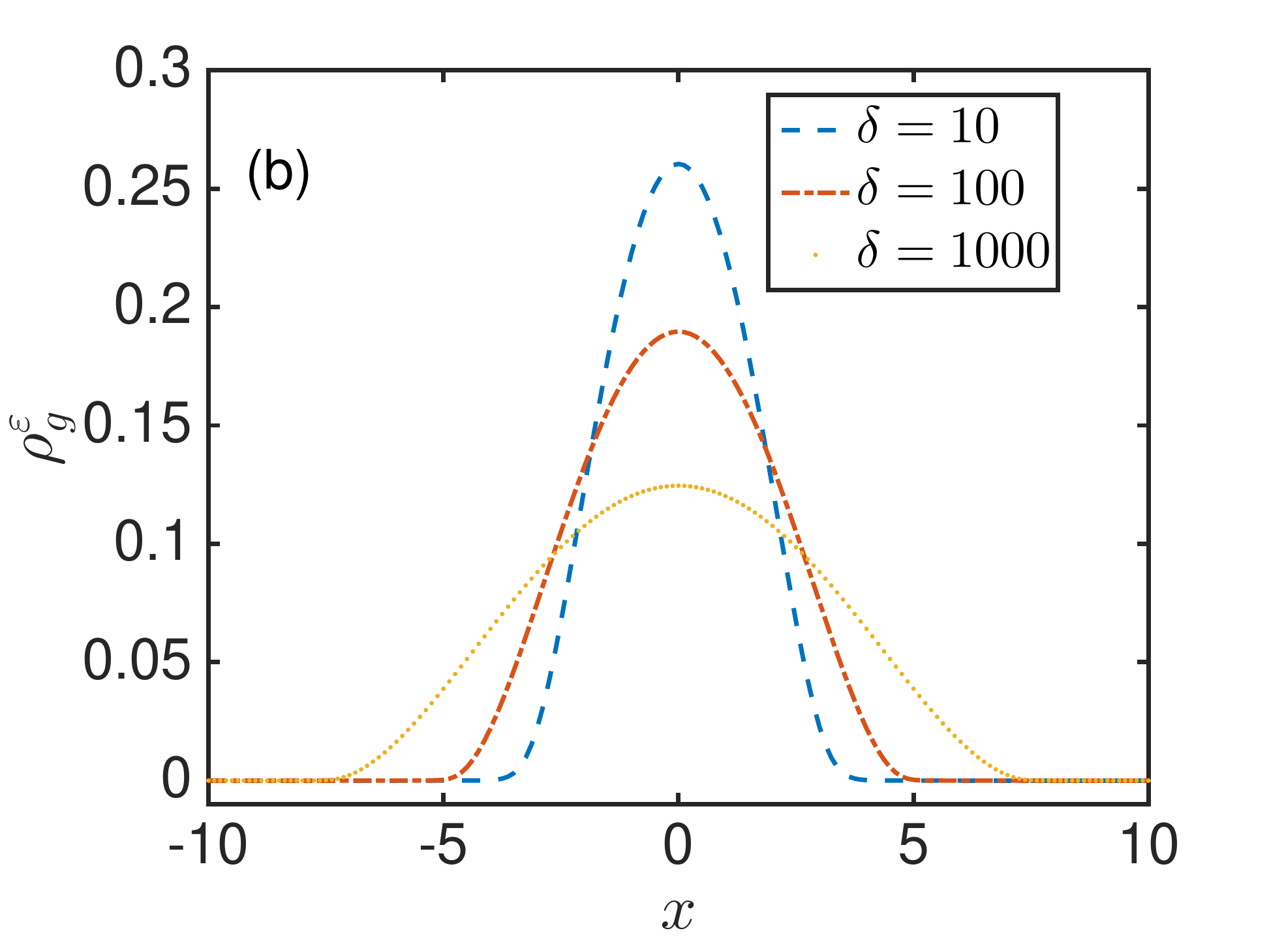,height=5cm,width=7cm,angle=0}
}
\centerline{
\psfig{
figure=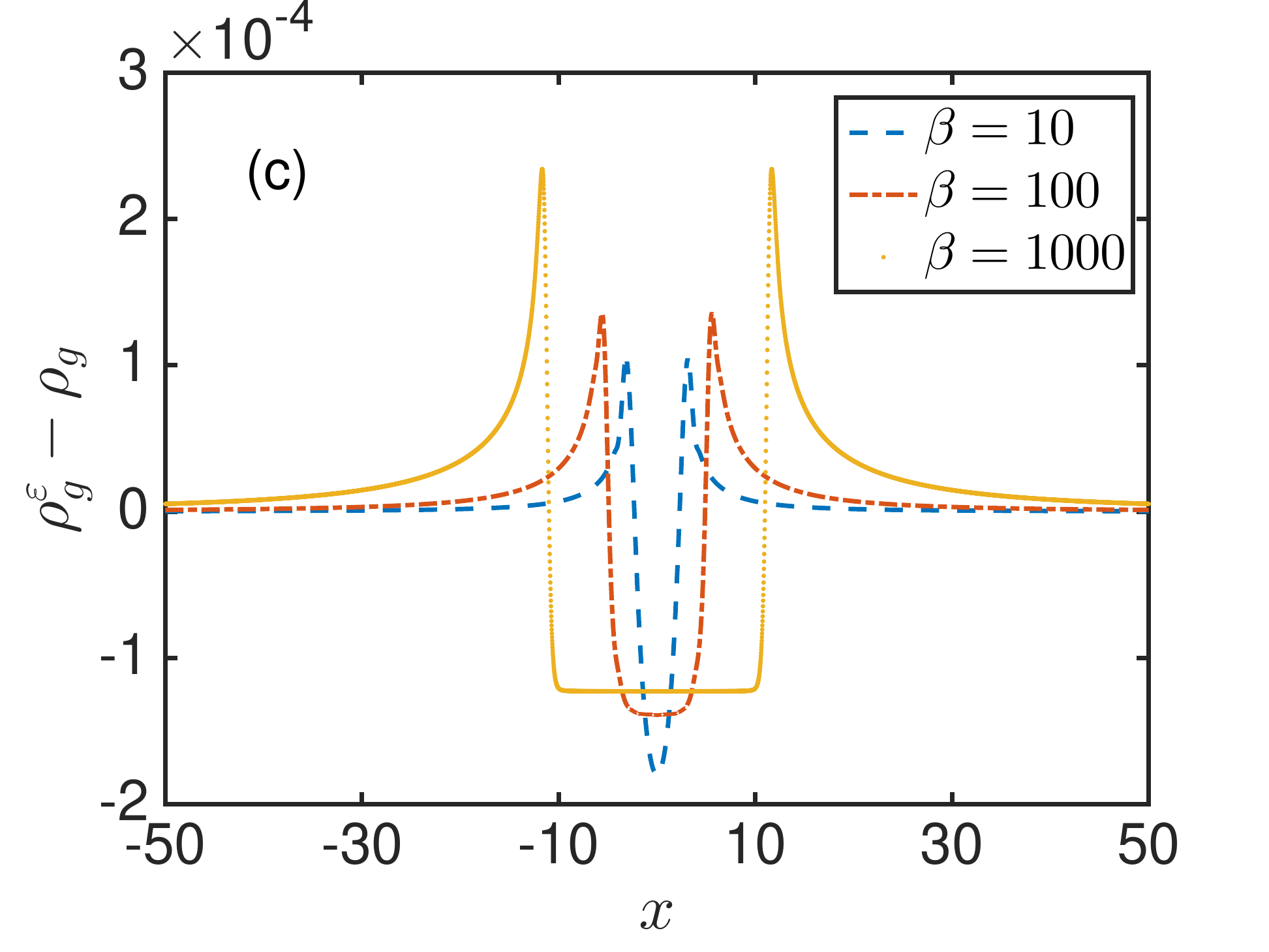,height=5cm,width=7cm,angle=0}
\psfig{
figure=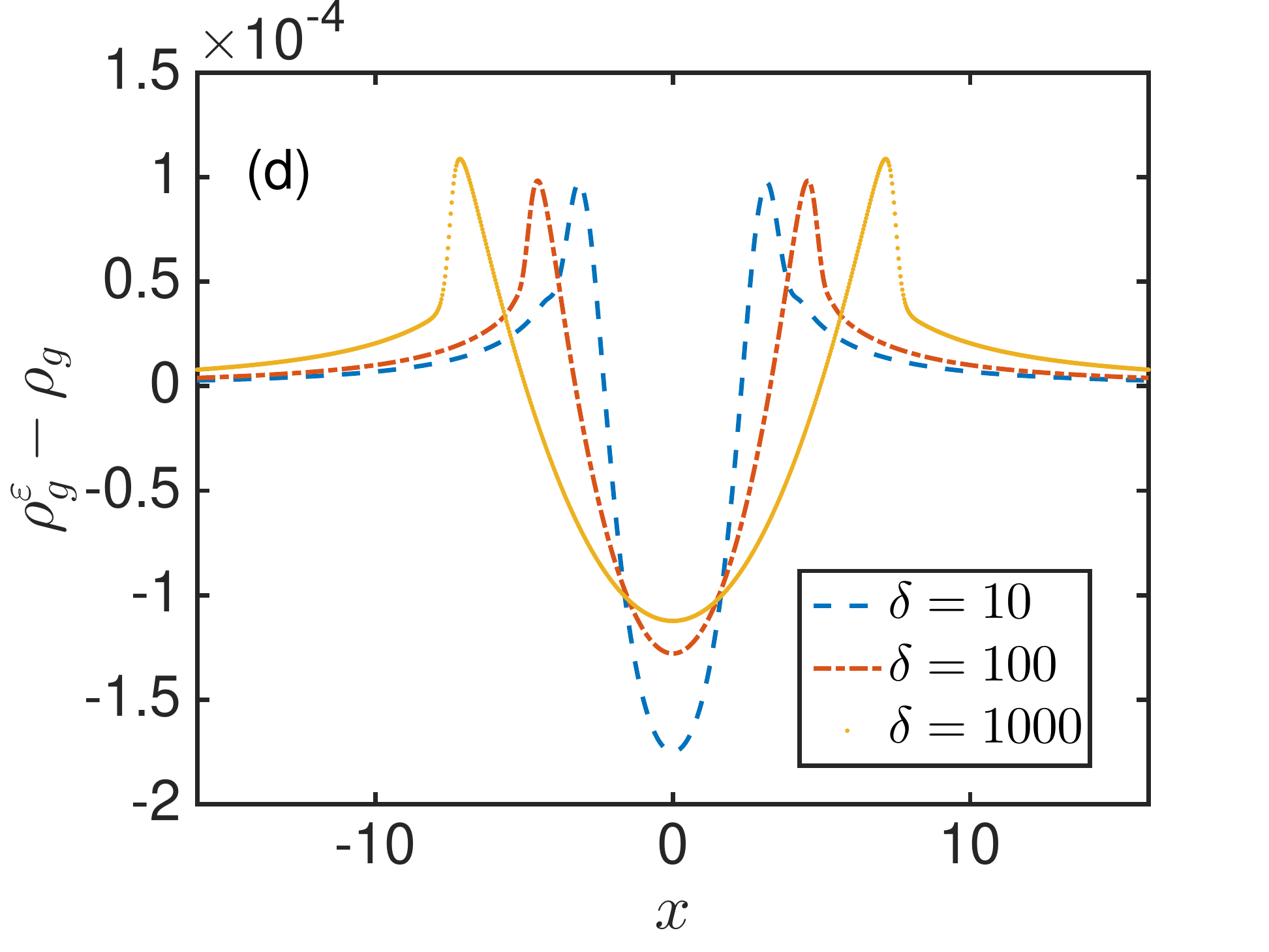,height=5cm,width=7cm,angle=0}
}
\caption{Ground state density $\rho_g^\vep$ \eqref{ground:reg} (top row, i.e. (a), (b)) and its comparison with 
$\rho_g$ \eqref{ground:ori} (bottom row, i.e. (c), (d)) with $\delta=10$ and different $\beta$'s (left column, i.e. (a), (c)) or $\beta=10$ and different $\delta$'s (right column, i.e. (b), (d)). Here we choose $\vep=10^{-4}$.}
\label{fig:ground_1D}
\end{figure}

\begin{table}
  \centering
  \begin{tabular}{ | p{0.7cm}| *{3}{p{0.9cm}|} p{0.9cm} || *{4}{p{0.9cm}|}}
         \cline{2-9}
    	\multicolumn{1}{c|}{} & \multicolumn{4}{c||}{rDF-APG} & \multicolumn{4}{c|}{Regularized Newton method}\\
	\hline
	\diagbox[width=1.1cm,height=0.7cm]{$\delta$}{$\beta$} &$10$ & $10^2$ & $10^3$ & $10^4$ &$10$ & $10^2$ & $10^3$ & $10^4$\\
		\hline
	10 & 18.34s & 12.39s & 4.29s & 1.67s & 4.54s & 3.76s & 3.47s & 3.85s \\
			\hline
	$10^2$ & 10.66s & 8.66s & 3.78s & 1.58s & 18.79s  & 13.42s & 9.36s & 7.62s \\
			\hline
	$10^3$ & 5.25s &  6.29s & 3.66s & 1.67s & 116.65s  & 79.41s & 45.49s & 34.03s \\
			\hline
	$10^4$ & 3.31s & 3.60s & 3.03s & 1.66s & 224.05s & 222.71s & 224.89s & 144.14s \\
    \hline
  \end{tabular}
  \caption{CPU time by using the rDF-APG and the regularized Newton methods. In the numerical test, we choose  $h=\frac{1}{64}$ and the computation domain to be $(-32,32)$. For the rDF-APG method, we choose $\vep=10^{-4}$, $tol=10^{-8}$ and apply the two-grid technique to speed up. All numerical experiments are performed on the same laptop. The CPU time shown in the table is accurate to the first digit. }
  \label{tab:CPU_cmp}
\end{table}

\subsection{Numerical examples in high dimensions}
The rDF-APG methd can be easily adapted to multi-dimensional problems with tensor grids.
In this section, we show examples in two dimensions (2D) and three dimensions (3D) to indicate the feasibility of our method for multi-dimensional problems.

Consider a box potential
\be \label{box_potential}
V(\bx)=\left\{\begin{array}{ll}
0,\quad \bx\in \Omega\subset {\mathbb R}^d,\\
\infty,\quad \bx\in \Omega^c.\\
\end{array}\right.
\ee

In 2D,  we take $\Omega=(-1,1)^2$.
The numerical ground state densities with $\vep=10^{-4}$ and different choices of $\beta$ and $\delta$ are shown in Figure \ref{fig:ground_2D}.
We observe the ground state becomes almost a constant inside the domain when $\beta\gg1$, while the ground state converges to another pattern  when $\delta\gg1$.

\begin{figure}[htbp]
\centerline{
\psfig{figure=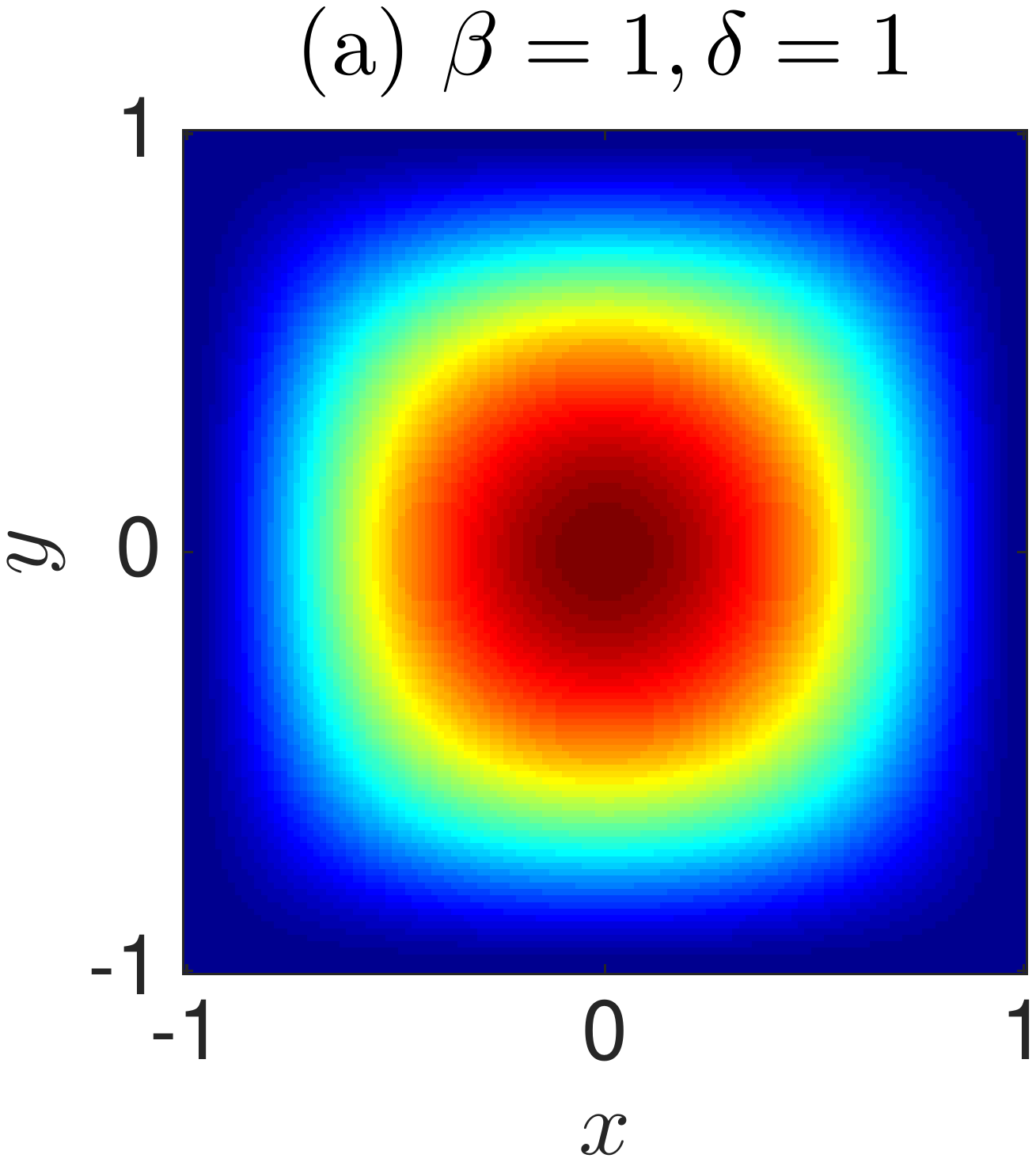,height=5.5cm,width=5.3cm,angle=0}
\psfig{figure=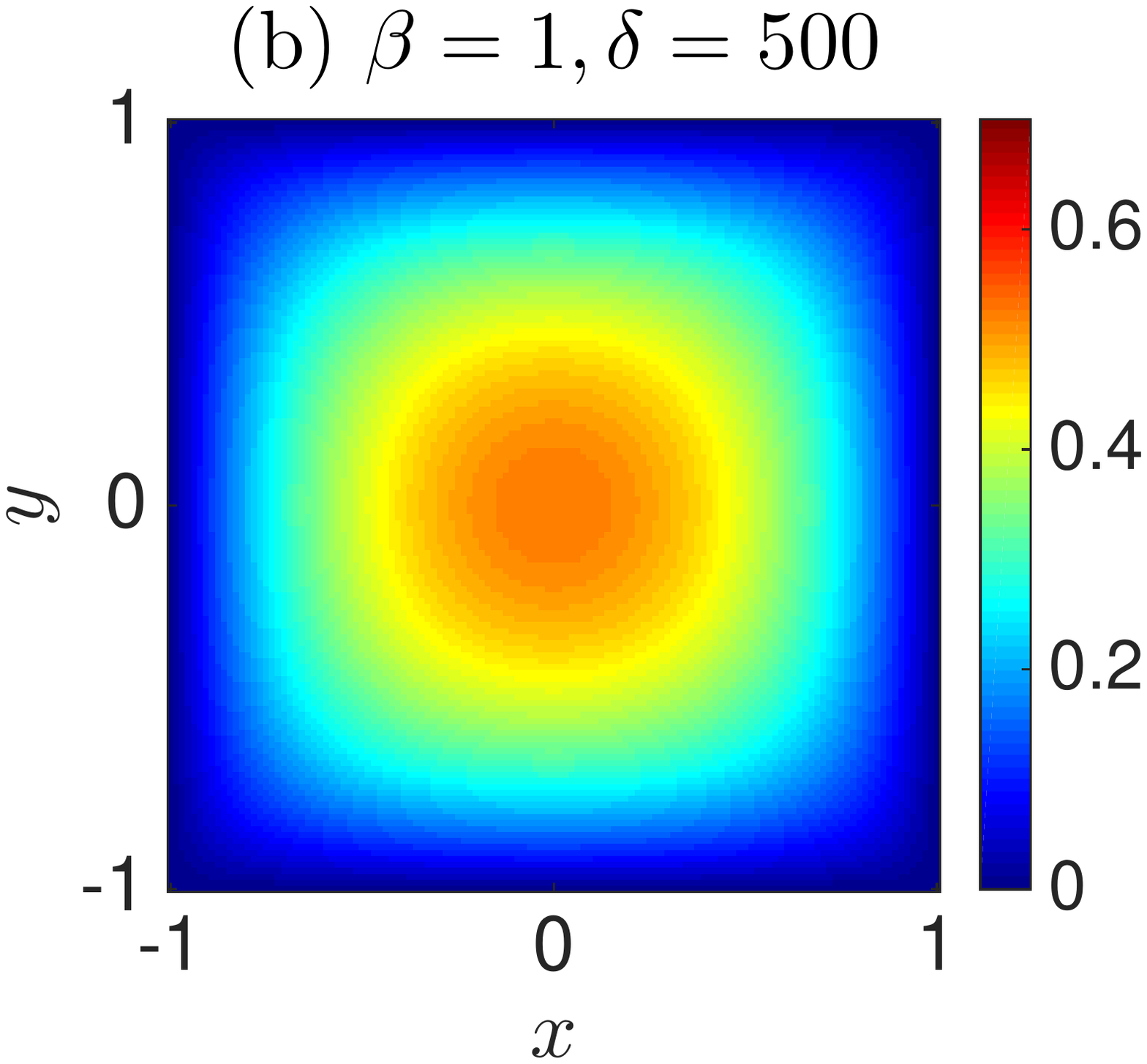,height=5.5cm,width=6cm,angle=0}}
\centerline{
\psfig{figure=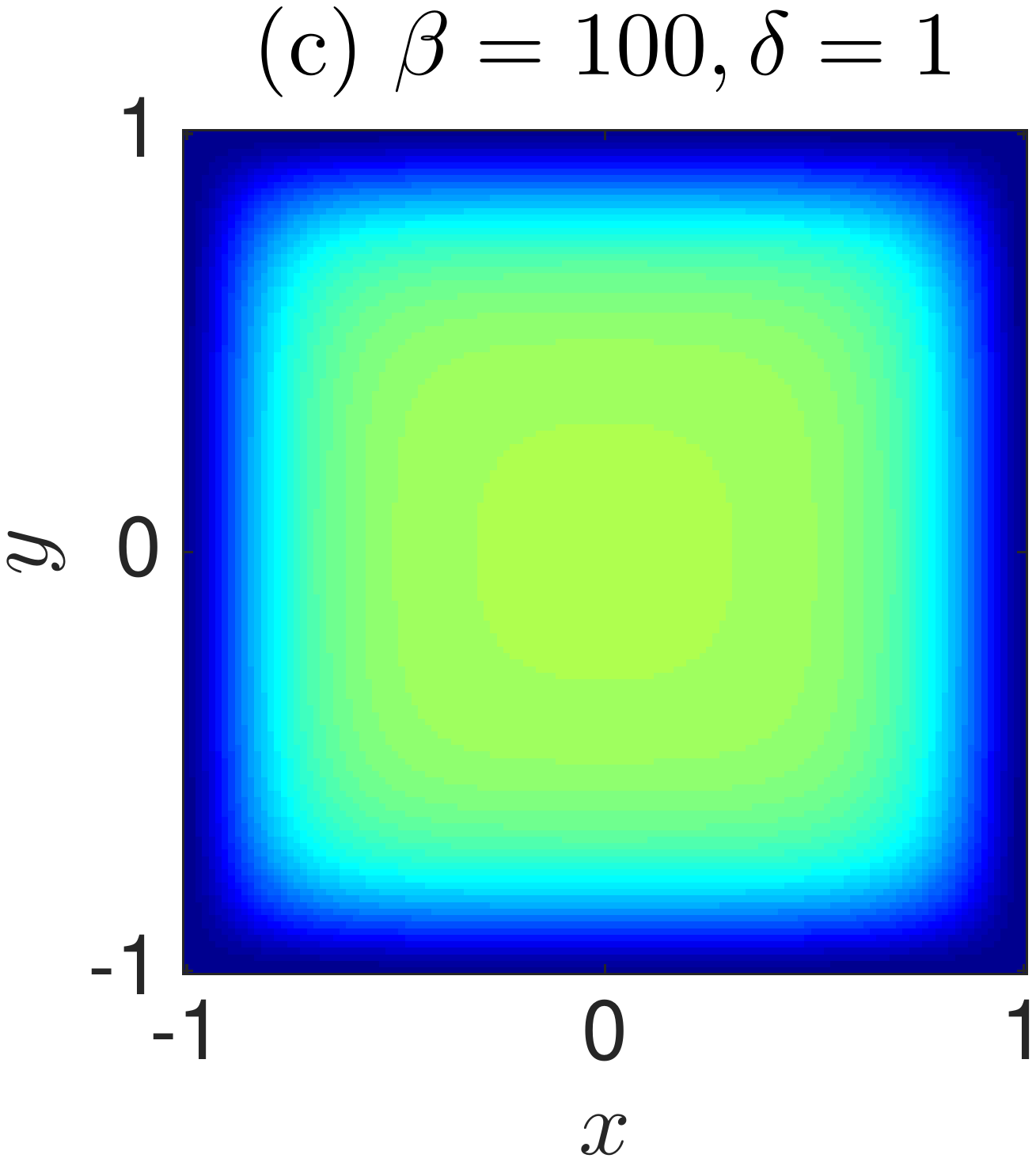,height=5.5cm,width=5.3cm,angle=0}
\psfig{figure=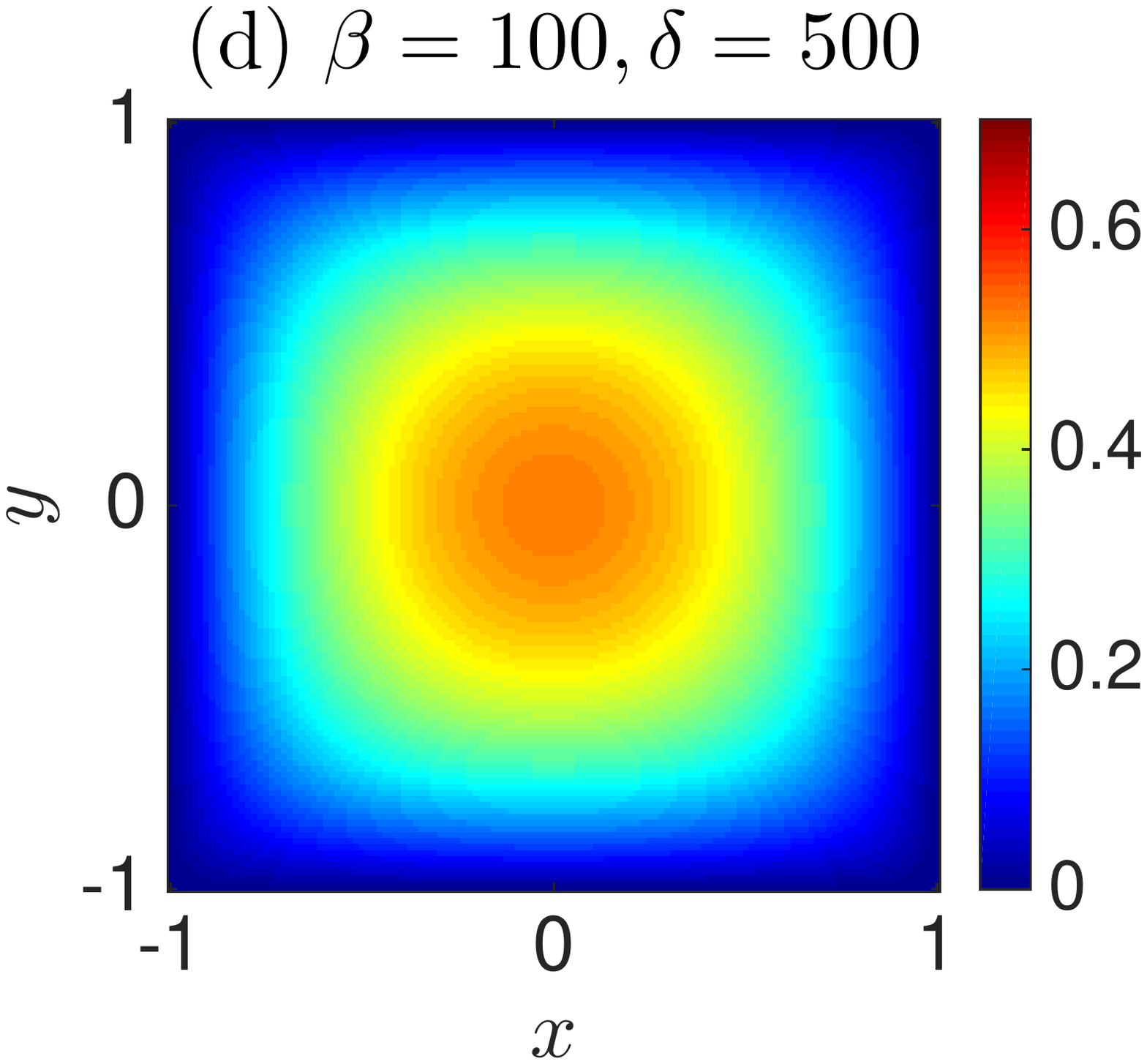,height=5.5cm,width=6cm,angle=0}}
\caption{Ground state densities $\rho_g^{\vep}(x,y)$ with $\beta=1,100$ (from top to bottom) and $\delta=1,500$ (from left to right) under the box potential in $\Omega=(-1,1)^2$. Here we choose $\vep=10^{-4}$.}
\label{fig:ground_2D}
\end{figure}

In 3D, we take  $\Omega=(-1,1)^3$.
Figure \ref{fig:ground_3D_Box} shows the isosurface end-cap geometry of the ground state with different values of $\beta$ and $\delta$.
As shown in the figure, again we can conclude that the solution spreads in a similar way as the 2D case shown in Figure \ref{fig:ground_2D} when $\beta$ or $\delta$ increases since the slices in Figure \ref{fig:ground_3D_Box} shows a similar pattern.

Next, we consider the whole space problem under an optical potential,
\be\label{potential:opt_3D}
V(x,y,z)=\frac{x^2+4y^2+4z^2}{2}+A(\sin^2(x)+\sin^2(y)+\sin^2(z)).
\ee
Obviously, when $A=0$, we get the harmonic potential.
Figure \ref{fig:ground_3D} shows the iso-surface of the numerical ground state densities with $\vep=10^{-4}$, $A=20$ and different values of $\beta$ and $\delta$.
The ground state densities shown in Figure \ref{fig:ground_3D} are consistent with the results shown in \cite{Ruan_gradientflow}.
Again we observe that the ground state densities become less concentrated as $\beta$ or $\delta$ increases.

\begin{figure}[htbp]
\centerline{
\psfig{figure=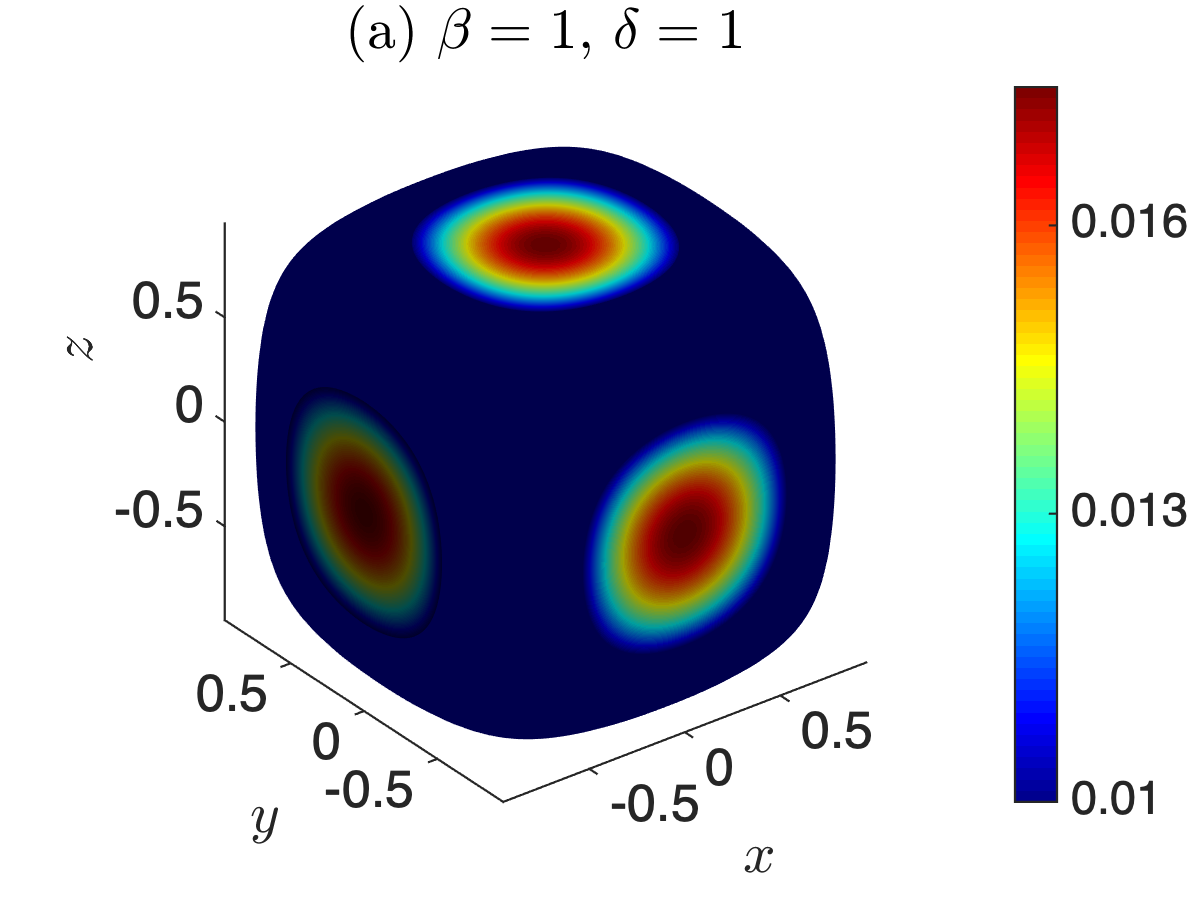,height=5cm,width=7cm,angle=0}
\psfig{figure=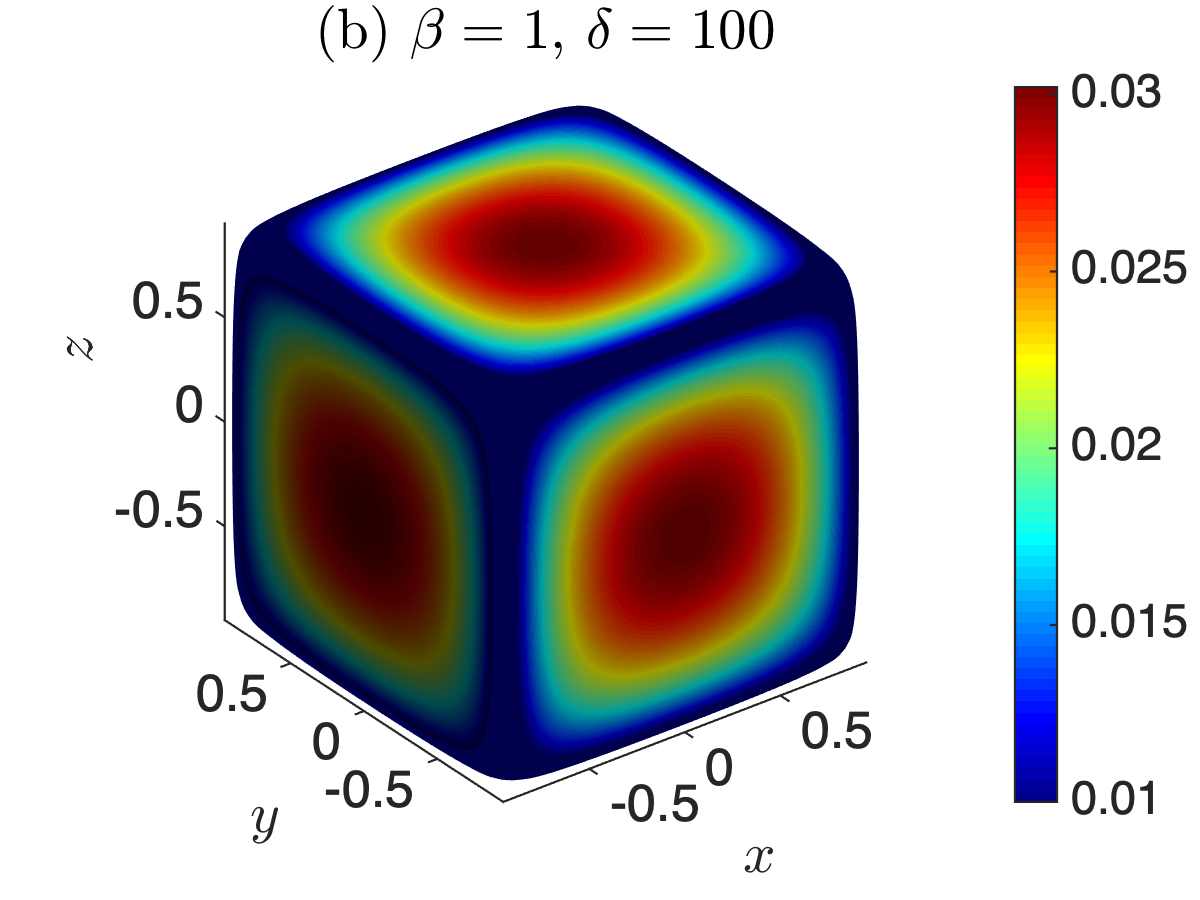,height=5cm,width=7cm,angle=0}}
\centerline{
\psfig{figure=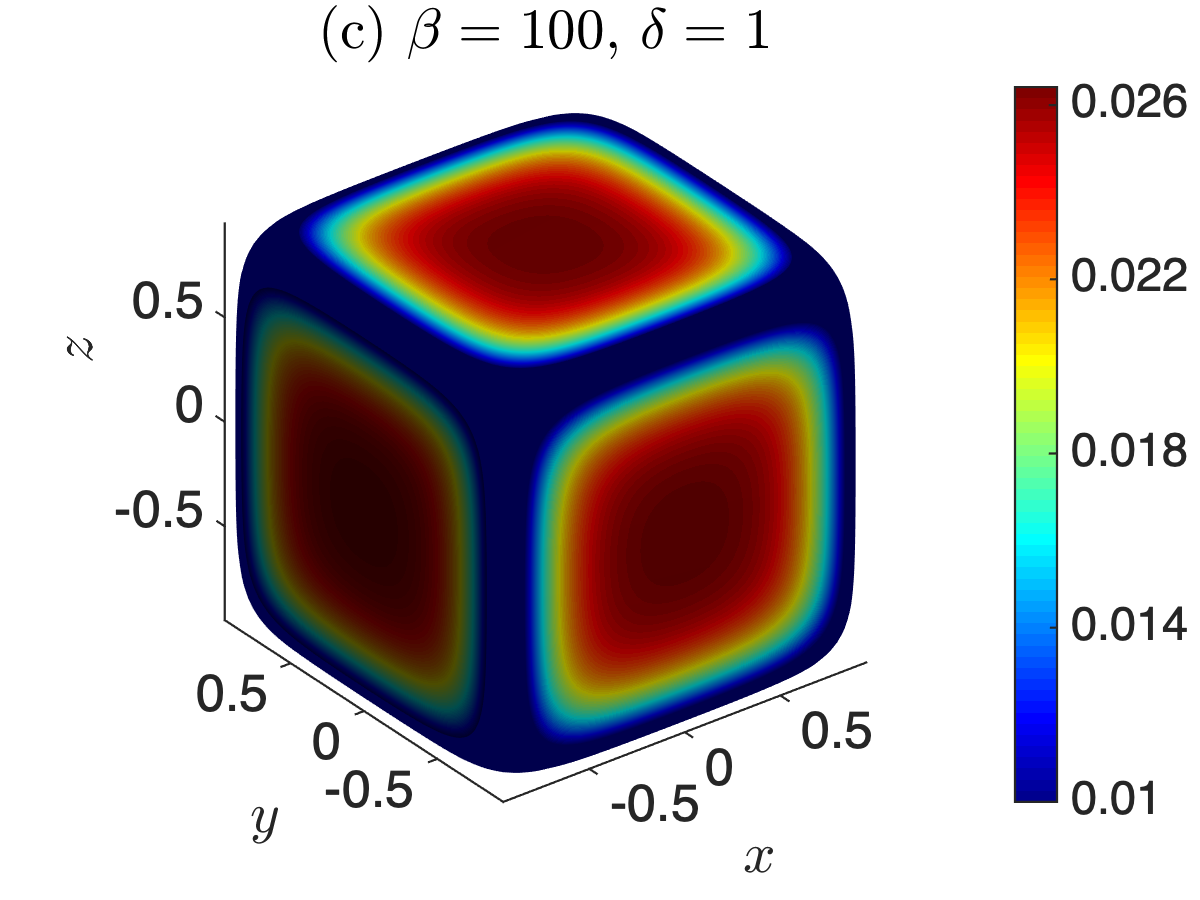,height=5cm,width=7cm,angle=0}
\psfig{figure=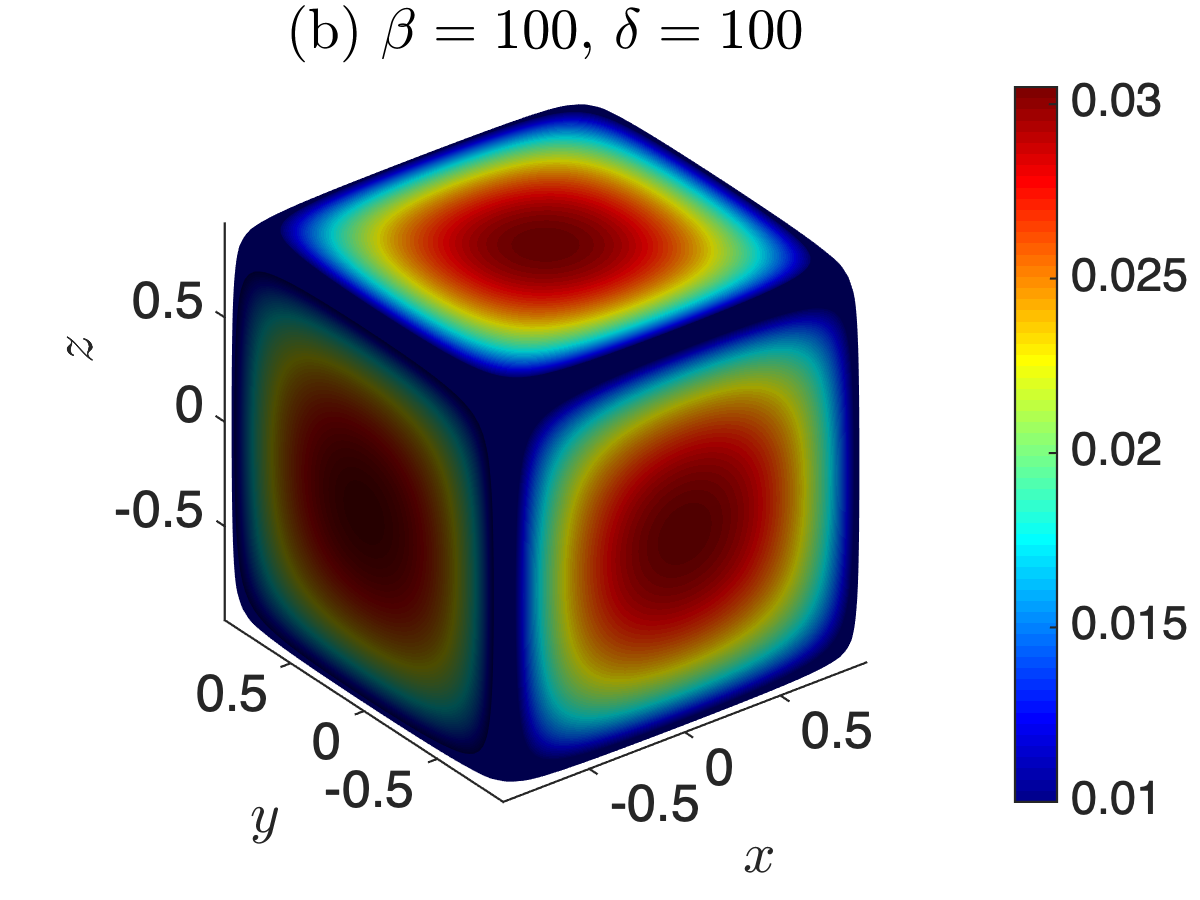,height=5cm,width=7cm,angle=0}}
\caption{Isosurface (isovalues $0.01$) of the ground state densities  $\rho_g^{\vep}(x,y,z)$ with an end-cap.
The external potential is chosen to be a box potential in $(-1,1)^3$ and  we choose $\vep=10^{-4}$.}
\label{fig:ground_3D_Box}
\end{figure}

\begin{figure}[htbp]
\centerline{
\psfig{figure=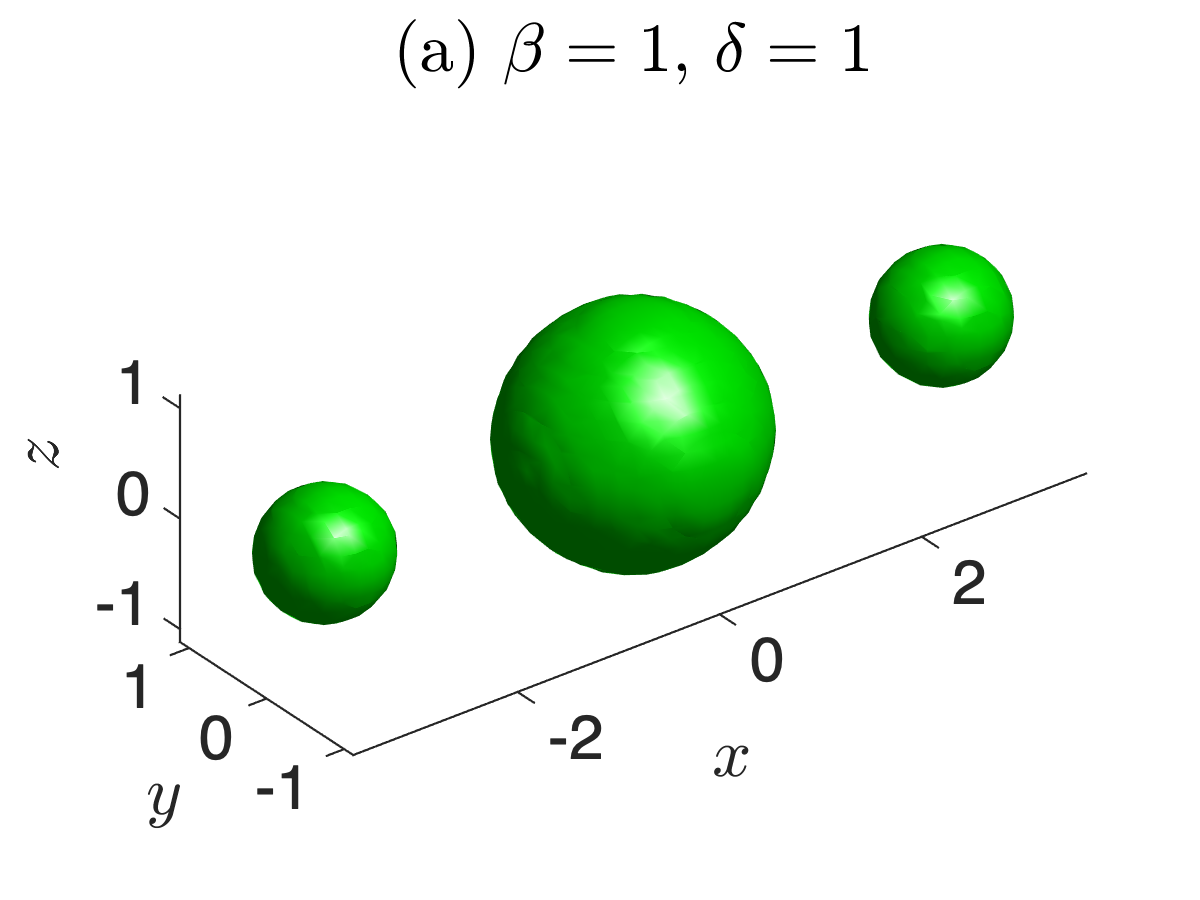,height=5cm,width=7cm,angle=0}
\psfig{figure=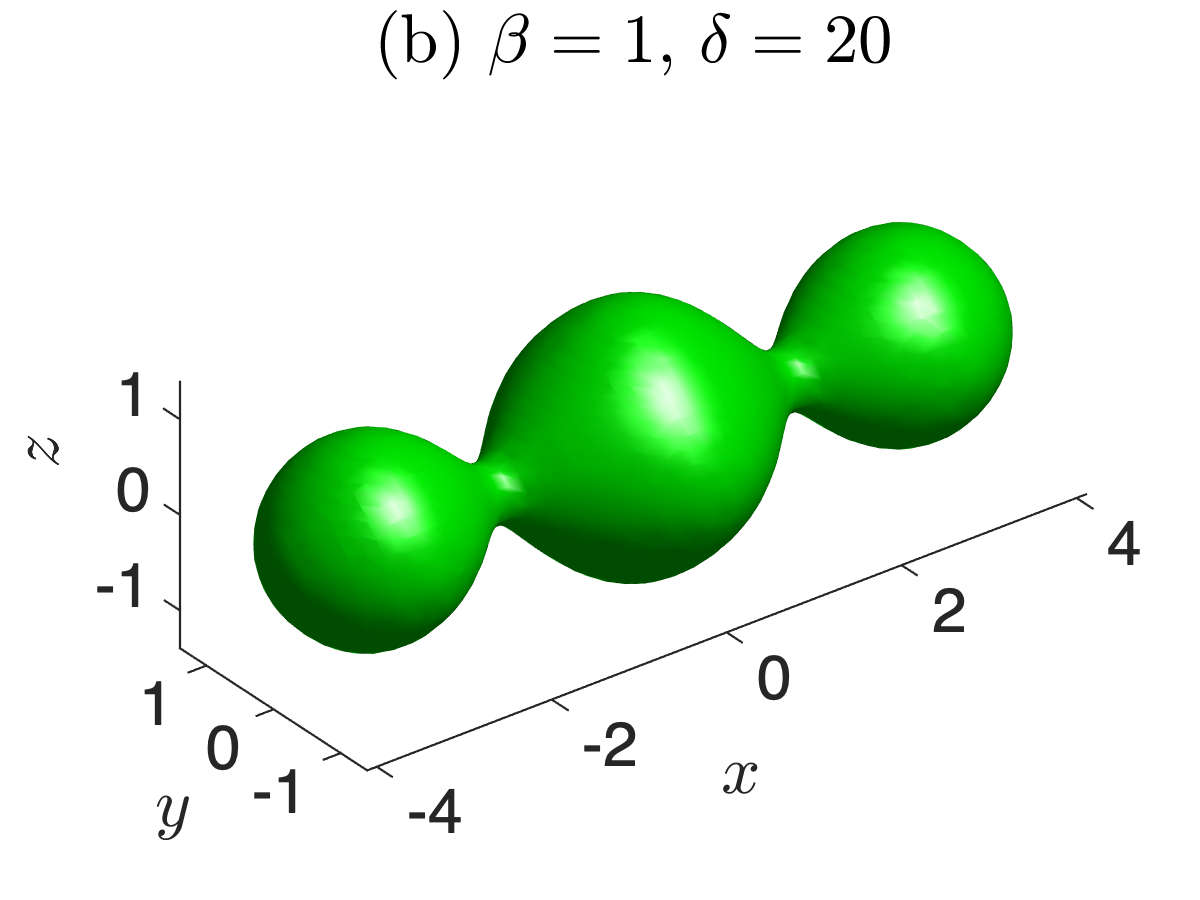,height=5cm,width=7cm,angle=0}}
\centerline{
\psfig{figure=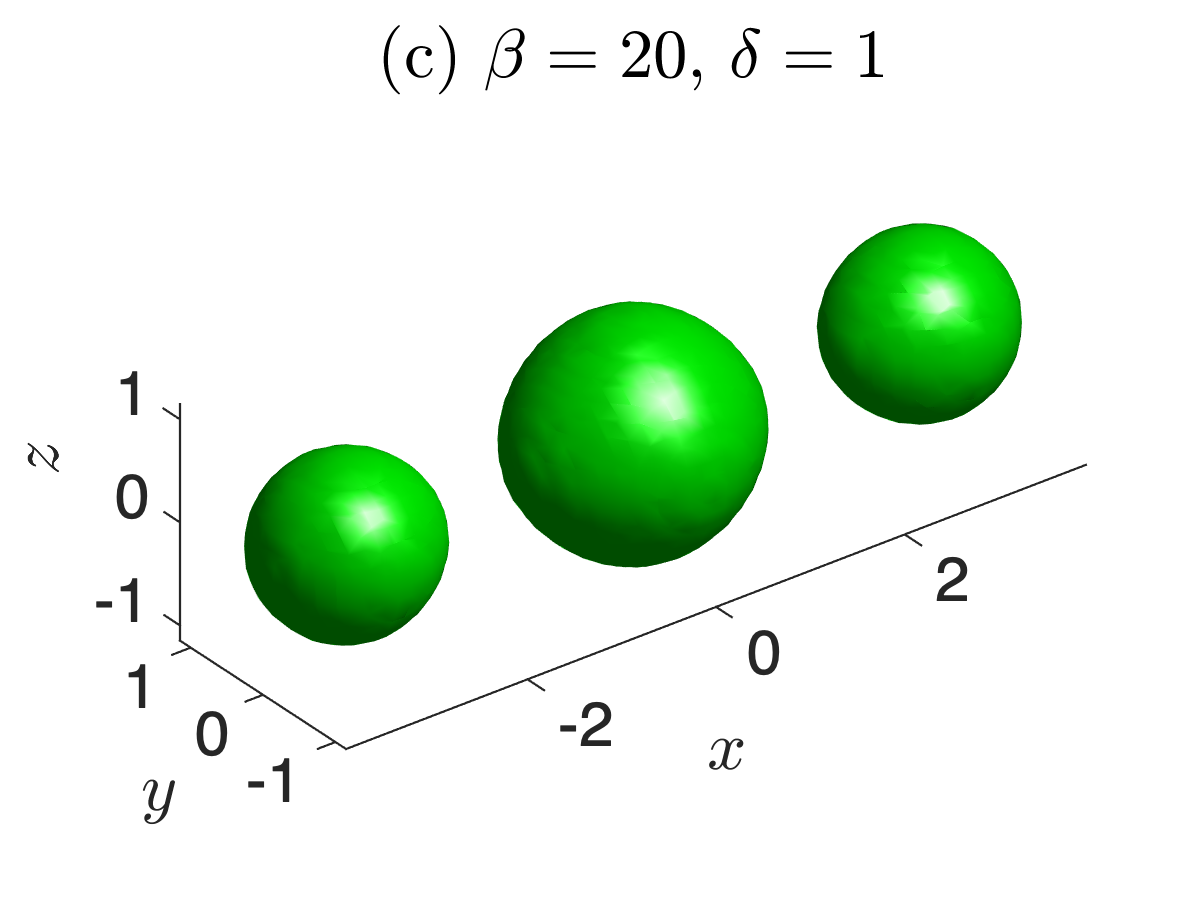,height=5cm,width=7cm,angle=0}
\psfig{figure=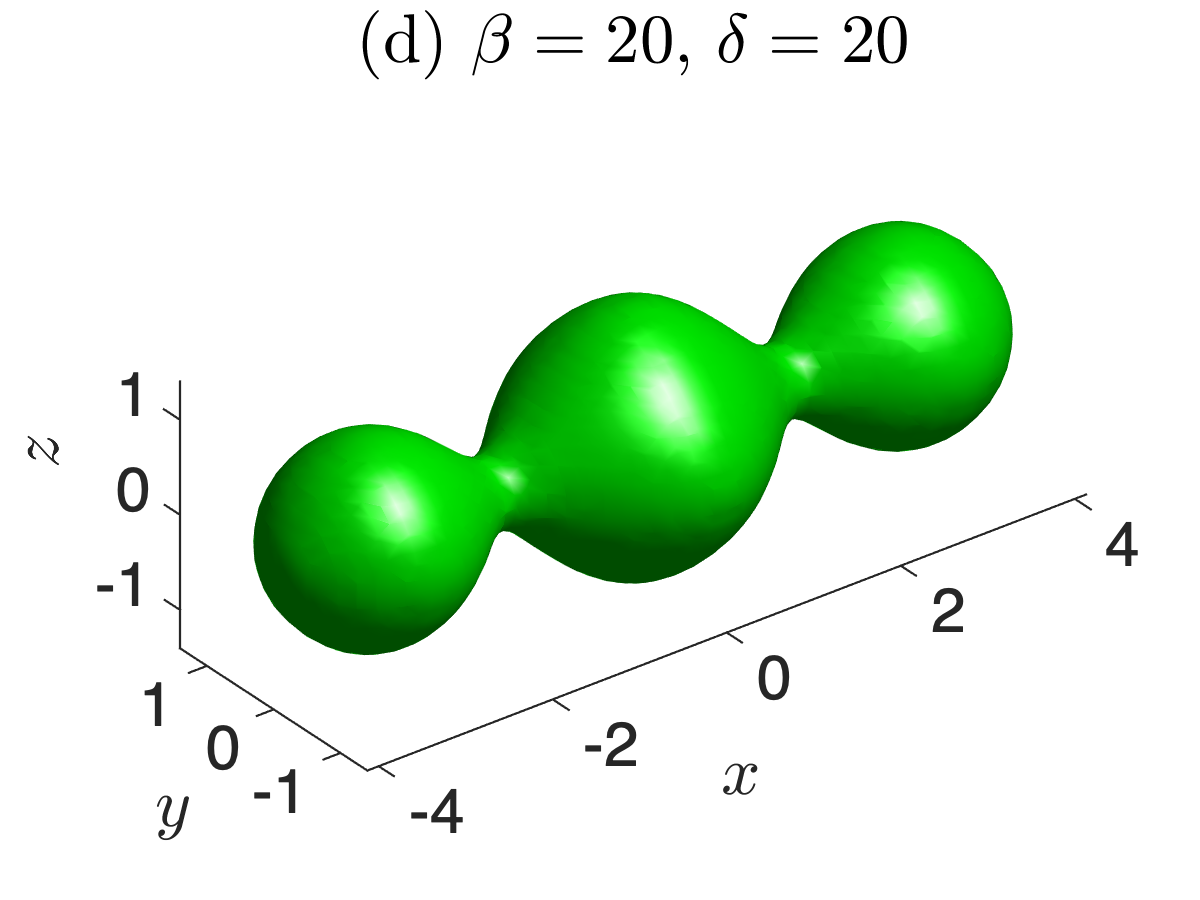,height=5cm,width=7cm,angle=0}}
\caption{Isosurface of the numerical  ground state densities  $\rho_g^{\vep}(x,y,z)$ with isovalues $0.01$.
The optical potential \eqref{potential:opt_3D} is considered with $A=20$ and different values of $\beta$ and $\delta$. Here we choose $\vep=10^{-4}$.}
\label{fig:ground_3D}
\end{figure}

\section{Conclusion}\label{sec:conclusion}
We proposed the rDF-APG method,
which computes the ground state density of the modified Gross-Pitaevskii equation (MGPE) by directly minimizing the regularized energy functional
formulated via the density function
and discretized via the finite difference method.
The convergence of the ground state densities as the regularization effect vanishes was rigorously proved.
The detailed finite difference discretization was provided,
with its spatial accuracy  analyzed.
Numerical experiments verified the spatial accuracy of the numerical ground state density computed via the rDF-APG method as well as  its convergence as $\vep\to0^+$.
Furthermore, we studied the effect of the interaction strength on the ground state density as well as on the performance of the rDF-APG method.
In particular, we showed the advantage of our method concerning the computational cost when the interaction, especially the HOI effect, is strong.

\appendix

\section{A finite difference discretization of $\hat{E}^{\vep}(\cdot)$} \label{appendix:reg_V_FD}
Numerical experiments suggest that a further regularization of the external potential term should be applied, and the new regularized energy functional $\hat{E}^{\vep}(\cdot)$ \eqref{energy_regularizeV} has a much better numerical performance than $E^{\vep}(\cdot)$ \eqref{energy_regularize}.
Denote $\hat{E}_{h}^{\vep}(\cdot)$ to be the finite difference discretization of $\hat{E}^{\vep}(\cdot)$.
For simplicity, only the 1D case is considered here.
Following a similar procedure as in \eqref{discretization:FD_derive_end}, we have
\be\label{Eh_V_rho}
\hat{E}_{h}^{\vep}(\brho_{h})=h\sum_{j=0}^{N-1}\left[\frac{1}{4}\frac{|\delta_x^+\rho_j|^2}{|\rho_j|+|\rho_{j+1}|+2\vep}+V(x_j)(\sqrt{\rho_j^2+\vep^2}-\vep)+\frac{\beta}{2}\rho_j^2+\frac{\delta}{2}\left|\delta_x^+\rho_{j}\right|^2\right],
\ee
where $\rho_{j+\frac12}:=(\rho_j+\rho_{j+1})/2$.
Similarly, we can compute its gradient. Denote
\be
\hat{\mathbf{g}}_{h}^{\vep}(\brho_{h})=( \frac{\p \hat{E}_{h}^{\vep}}{\p\rho_1},\frac{\p \hat{E}_{h}^{\vep}}{\p\rho_2},\dots,\frac{\p \hat{E}_{h}^{\vep}}{\p\rho_{N-1}})^T.
\ee
Then its $j$-th component can be explicitly computed  as
\be\label{Eh_V_rho_gradient}
\hat{\mathbf{g}}_{h}^{\vep}[j]=h\left[-\frac{\delta_x^+\tf_{j-1}}{2}-\frac{\tf_{j-1}^2+\tf_j^2}{4}s_j+\frac{V(x_j)\rho_j}{\sqrt{\rho_j^2+\vep^2}}+\beta \rho_j-\delta(\delta_x^2\rho_j)\right],
\ee
where $s_j$ and $\tf_j$ were defined in \eqref{def:s_f_APG}.
Comparing $\hat{\mathbf{g}}_{h}^{\vep}[j]$ \eqref{Eh_V_rho_gradient} and $\tilde{\mathbf{g}}_h^{\vep}[j]$ \eqref{Eh_rho_gradient_APG}, we observe that
$\hat{\mathbf{g}}_{h}^{\vep}[j]$ is more regular than
$\tilde{\mathbf{g}}_h^{\vep}[j]$ where $\rho_j\approx0$  and $V(x_j)\gg1$.
The observation could possibly explain  the better performance of the new formulation $\hat{E}_{h}^{\vep}(\cdot)$ \eqref{Eh_V_rho}.

Here we define $\hat{\brho}_{g,h}^{\vep}\in\mathbb{R}^{N-1}$ as follows,
\be
\hat{\brho}_{g,h}^{\vep}=\argmin_{\brho_{h}\in W_h} \hat{E}_{h}^{\vep}(\brho_{h}).
\ee
Obviously, $\hat{\brho}_{g,h}^{\vep}$ is the numerical approximation of $\hat{\rho}_g^{\vep}$ \eqref{ground:regV}.
Then it can be shown that Theorem \ref{thm:conv_h_E}, Theorem \ref{thm:conv_h} and Theorem \ref{thm:conv_h_E_discrete} still hold true  if we replace $\brho_{g,h}^{\vep}$ by $\hat{\brho}_{g,h}^{\vep}$ in the theorems.
The argument is almost exactly the same and thus omitted here for brevity.

\bibliographystyle{siamplain}

\end{document}